\documentclass[final]{IEEEtran}

\usepackage{stfloats}
\usepackage{amsfonts}
\usepackage{amssymb}
\usepackage{amsthm}
\usepackage{cite}
\usepackage[cmex10]{amsmath}
\usepackage{float}
\usepackage{color}
\usepackage{stfloats,fancyhdr}
\usepackage{amsmath,bm}
\usepackage{algorithm}
\usepackage{algorithmic}
\usepackage{multirow}
\usepackage{changepage}
\usepackage[normalem]{ulem}
\usepackage{amsthm}
\usepackage{balance}

\newtheorem{theorem}{Theorem}
\newtheorem{lemma}{Lemma}
\newtheorem{proposition}{Proposition}
\newtheorem{corollary}{Corollary}
\newtheorem{definition}{Definition}
\newtheorem{assumption}{Assumption}
\newtheorem{remark}{Remark}
\newtheorem{example}{Example}

\ifCLASSINFOpdf
\usepackage[pdftex]{graphicx}
\DeclareGraphicsExtensions{.pdf,.jpeg,.png}
\else
\usepackage[dvips]{graphicx}
\DeclareGraphicsExtensions{.eps}
\fi

\usepackage{subfigure}
\usepackage{fancybox,dashbox}
\usepackage{authblk}
\usepackage{tabularx}

\newcommand{\myexpect}[1]{\mathsf{E}\left[#1\right]}
\newcommand{\myrank}[1]{\text{Rank}\left(#1\right)}
\newcommand{\myprob}[1]{\mathsf{Prob}\left[#1\right]}

\begin{document}
		
\title{Remote State Estimation with Smart Sensors over Markov Fading Channels}

\author{Wanchun Liu, Daniel E.\ Quevedo, Yonghui Li,  
	Karl Henrik Johansson 
	and Branka Vucetic
}

\maketitle

\begin{abstract}
\let\thefootnote\relax\footnote{W. Liu, Y. Li and B. Vucetic are with School of Electrical and Information Engineering, The University of Sydney, Australia.
	Emails:	\{wanchun.liu,\ yonghui.li,\ branka.vucetic\}@sydney.edu.au. 
D. E. Quevedo is with the School of Electrical Engineering and Robotics, Queensland University of Technology (QUT), Brisbane, Australia.	Email: daniel.quevedo@qut.edu.au.
K. H. Johansson is with School of Electrical Engineering and Computer Science, KTH Royal
Institute of Technology, Stockholm, Sweden. Email: kallej@kth.se.	
}
We consider a fundamental remote state estimation problem of discrete-time linear time-invariant (LTI) systems. A smart sensor forwards its local state estimate to a remote estimator over a time-correlated multi-state Markov fading channel, where the packet drop probability is time-varying and depends on the current fading channel state.
We establish a necessary and sufficient condition for mean-square stability of the remote estimation error covariance in terms of the state transition matrix of the LTI system, the packet drop probabilities in different channel states, and the transition probability matrix of the Markov channel states.
To derive this result, we propose a novel estimation-cycle based approach, and provide new element-wise bounds of matrix powers.  
The stability condition is verified by numerical results, and is shown more effective than existing sufficient conditions in the literature. We observe that the stability region in terms of the packet drop probabilities in different channel states can either be convex or non-convex depending on the transition probability matrix of the Markov channel states. 
Our numerical results suggest that the stability conditions for remote estimation may coincide for setups with a smart sensor and with a conventional one (which sends raw measurements to the remote estimator), though the smart sensor setup achieves a better estimation performance.
\end{abstract}


\section{Introduction}	
	
\subsection{Motivation}
In the long-term evolution of wireless applications from conventional wireless sensor networks (WSNs) to the Internet-of-Things (IoT) and the Industry 4.0, remote estimation is a key component~\cite{WSN,IoTestimation,Indusrtial40,Antonakoglou2018towards}.
Driven by Moore's Law, 
the accelerated development and adoption of smart sensor technology enables low-cost sensors with high computational capability~\cite{moore}.
Thus, in a number of remote estimation applications, it is practical to use smart sensors (e.g., with Kalman filters) to pre-estimate the dynamic states, and then send the estimated states rather than the raw measurements to the remote estimator. In the presence of communication constraints, the smart sensors provide better estimation performance than the conventional sensors that purely send raw measurement data to the remote estimator~\cite{schenato2008optimal}.

Unlike wired communications, wireless communications are unreliable and the channel status varies with time due to multipath propagation and shadowing caused by obstacles affecting the wave propagation.
The transition process of the fading channel states is usually modeled as a Markov process~\cite{Parastoo,markovchannel1,markovchannel2}, and different channel states lead to different packet drop probabilities of  transmissions.
The presence of an unreliable wireless communication channel degrades the estimation performance, and in some cases even lead to instability.
Whilst stability when using conventional sensors has been well investigated, see literature survey below, stability when using a smart sensor has been much less considered.
In this paper, we tackle the fundamental problem: what are necessary and sufficient conditions on system parameters that ensure stochastic stability  of a smart-sensor-based remote estimation system over a Markov fading channel?

\subsection{Related Works}
The existing work on remote estimation can be divided into two categories based on  the sensor's computational capability.

In the \textbf{conventional sensor scenario}, the sensor sends raw measurements to the remote estimator.
When considering a static wireless channel, where neither the transceivers nor the wireless environment are moving, the packet drop probability during the remote estimation process remains fixed, so the packet arrival process is a Bernoulli process. It was proven in~\cite{estimation_first} that there exists a critical packet drop probability, such that the mean estimation error covariance is bounded for all initial conditions and diverges for some initial condition if the packet drop probability is less or greater than the critical probability, respectively.
This result was further extended to a scenario with random packet delays in~\cite{schenato2008optimal}.
By modeling the packet arrival process as a Markovian binary switching process, sufficient conditions for stability in the sense of peak covariance were obtained in \cite{XieXie,minyihuang}.
For situations where the number of consecutive packet dropouts constitutes a bounded Markov process,  peak covariance stability was investigated in \cite{xiao2009kalman}.
By modeling the sequence of packet dropouts as a stationary
finite-order Markov process,  a necessary and  sufficient stability condition was obtained in the sense of mean estimation error covariance in~\cite{youfu11}.
In contrast to~\cite{XieXie,minyihuang,xiao2009kalman,youfu11}, which directly model packet-dropouts as a Markovian process and abstract away the underlying wireless channel, 
Markovian fading channel states were explicitly considered in~\cite{Quevedo}. A sufficient condition for  exponential stability was derived by using stochastic Lyapunov functions.
With the same multi-state Markov channel model as~\cite{Quevedo}, optimal transmit power allocation policy under different channel conditions was proposed in~\cite{chakravorty2019remote} to achieve the minimum remote estimation error.

More recently, closed-loop control systems over multi-state Markov channels were investigated in~\cite{lun2020impact,KangJIoT,hu2019co}. Under an ideal assumption of perfect sensor measurements, a necessary and sufficient stability condition was obtained in~\cite{lun2020impact}, where the sensor and the controller were co-located; a sufficient stability condition of a half-duplex control system was obtained in~\cite{KangJIoT}, where the controller applied a scheduling policy determining when to receive the sensor's packet or to transmit a control packet to the actuator. 
In~\cite{hu2019co}, sufficient stability conditions in terms of maximum allowable transmission interval of a nonlinear system was investigated.
The work~\cite{minero2012stabilization} focused on bit-rate limited error-free communication channels, where the number of bits to be transmitted in each time slot formed a Markov chain. By combining results from quantization theory with insights from Markov Jump Linear Systems, \cite{minero2012stabilization} examined how the quantization errors induced by finite-bit quantizers affect the control and estimation quality.

In the \textbf{smart sensor scenario}, an estimator (e.g., based on a Kalman filter) of the sensor side pre-processes the raw measurements, such that an estimate is transmitted to the remote estimator over the wireless channel.
It has been rigorously proved in \cite{schenato2008optimal} that smart sensor scenario performs better
than the conventional sensor scenario when taking into account the transmission
delay and failures.
However, unlike the conventional-sensor-based scenario, most of the theoretical research on smart-sensor-based remote estimation considered static channels and assumed independent and identically distributed (i.i.d.) packet dropouts~\cite{schenato2008optimal,shi2012optimal,shi2012scheduling,wu2018optimal,leong2017sensor,LiSecurity,DingSecurity,AlexSecurity,Kang2019ICC}.
In \cite{schenato2008optimal}, a necessary and sufficient condition for remote estimation stability was derived in the mean-square sense.
In \cite{shi2012optimal}, an optimal sensor power scheduling policy under a sum power constraint was obtained.
In \cite{shi2012scheduling}, the optimal transmission scheduling policy of two sensors each measuring the state of one of the two systems was obtained in a closed form, where the sensors shared a single wireless channel.
This work was extended to a multi-sensor-multi-channel scenario in~\cite{wu2018optimal}, where the optimal transmission schedule policy was obtained by solving a Markov decision process problem.
In~\cite{leong2017sensor}, an optimal event-triggered transmission policy of a multi-sensor-multi-channel remote estimation system was proposed with a combined design target: the estimation error and the energy consumption of sensor transmissions.

In addition, optimal smart sensor transmission scheduling policies for single and multiple wireless channel scenarios were investigated under the presence of jamming attacks in~\cite{LiSecurity} and~\cite{DingSecurity}, respectively; an optimal transmission scheduling policy under the presence of an eavesdropper was proposed in~\cite{AlexSecurity} to minimize the remote estimation error at the dedicated receiver while keeping the eavesdropper's estimation error as large as possible.

More recently, a remote estimation system with retransmissions was proposed in~\cite{Kang2019ICC}, where the smart sensor can decide whether to retransmit the unsuccessfully transmitted local estimate (with a longer latency) or to send a new estimate (with a lower reliability). The obtained optimal retransmission scheduling policy found the optimal balance between the transmission latency and reliability on the remote estimation performance.

\subsection{Contributions}
In this paper,
we investigate mean-square stability of smart-sensor-based  remote estimation over an error-prone multi-state time-homogeneous Markov channel.
The $M$-state fading-channel model under consideration introduces an unbounded Markov chain in the analysis of the remote estimation system, which presents some non-trivial challenges.
The main contributions are summarized as below.
\begin{enumerate}
	\item We derive a necessary and sufficient condition on the stability of a remote state estimation system in terms of the system matrix $\mathbf{A}$, the packet drop probabilities in different channel states $\{d_1,\dots,d_M\}$ and the matrix of the channel state transitions $\mathbf{M}$. The remote state estimation is mean-square stable if and only if $\rho^2(\mathbf{A})\rho(\mathbf{DM})<1$, where $\rho(\cdot)$ denotes the spectral radius, and $\mathbf{D}$ is the diagonal matrix generated by $\{d_1,\dots,d_M\}$.
	\item We derive asymptotic upper and lower bounds of the estimation error function in terms of the number of consecutive packet dropouts $i$, which are in the same order of $(\rho(\mathbf{A})+\epsilon)^i$ and $\rho^i(\mathbf{A})$, respectively, where $\epsilon$ is an arbitrarily small positive number.
\end{enumerate}	
To obtain these results, we propose a novel estimation-cycle based analytical approach. Moreover, we further develop the asymptotic theory of matrix power, which provides new element-wise bounds of matrix powers.

\begin{figure*}[t]
	\centering\includegraphics[scale=0.7]{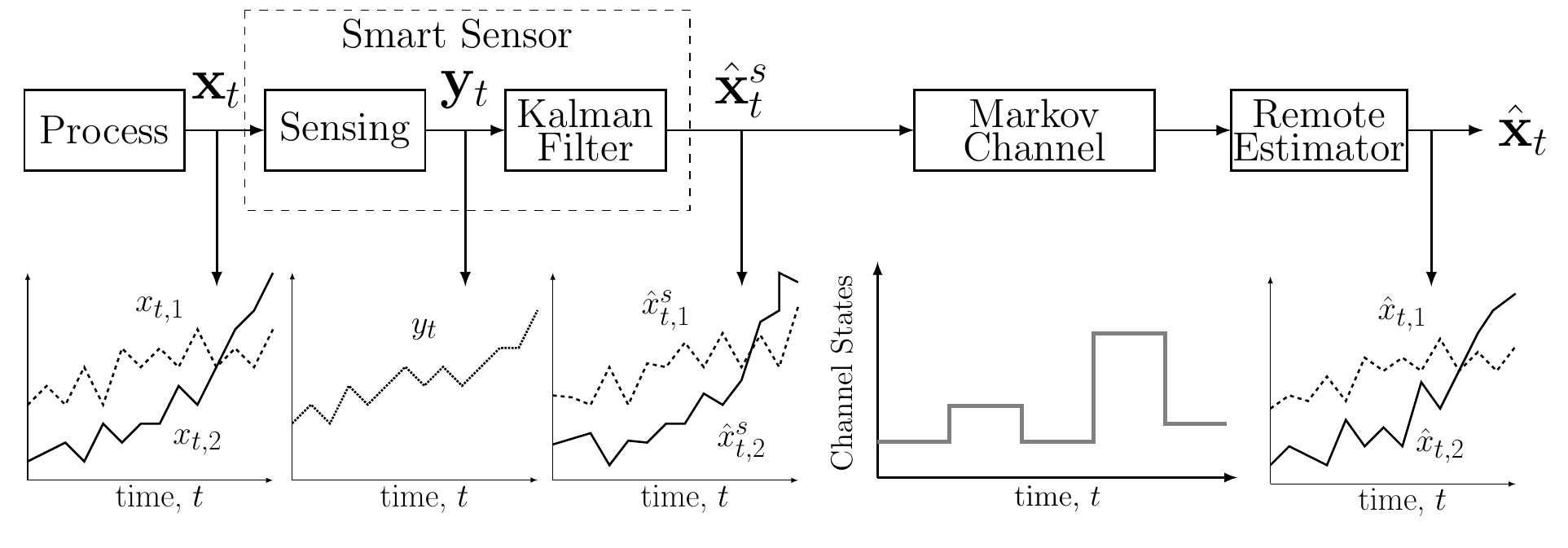}
	\vspace{-0.5cm}
	\caption{Remote state estimation system.}
	\label{retransmit_system_model}
\end{figure*}
\subsection{Outline and Notations}	
The remainder of this paper is organized as follows:
Section~\ref{sec:sys} presents the model of the remote estimation system using a Markov fading channel.
Section~\ref{sec:main_results} presents and discusses the main results of the paper.
Section~\ref{sec:proof} proposes a stochastic estimation-cycle based analysis approach and derives some element-wise bounds of matrix powers. They are used in 
Section~\ref{sec:main} to prove the main results.
Section~\ref{sec:num} numerically evaluates the performance of the remote estimation system, and verifies the theoretical results.
Section~\ref{sec:con} draws conclusions.

\emph{Notations:} 
Sets are denoted by calligraphic capital letters, e.g., $\mathcal{A}$.
$\mathcal{A} \backslash \mathcal{B}$ denotes set subtraction.
Matrices and vectors are denoted by capital and lowercase upright bold letters, e.g., $\mathbf{A}$ and $\mathbf{a}$, respectively.
$\vert \mathcal{A}\vert$ denotes the cardinality of the set $\mathcal{A}$.
 $\mathsf{E}\left[A\right]$ is the expectation of the random variable $A$.
$(\cdot)^{\top}$ is the matrix transpose operator. $\| \mathbf{v} \|_1$ is the sum of the vector $\mathbf{v}$'s elements. 
$|\mathbf{v}| \triangleq \sqrt{\mathbf{v}^\top \mathbf{v}}$ is the Euclidean norm of a vector $\mathbf{v}$.
$\text{Tr}(\cdot)$ is the trace operator. $\text{diag}\{v_1,v_2,...,v_K\}$ denotes the diagonal matrix with the diagonal elements $\{v_1,v_2,...,v_K\}$. $\mathbb{N}$ and $\mathbb{N}_0$ denote the sets of positive and non-negative integers, respectively.
$\mathbb{R}^m$ denotes the $m$-dimensional Euclidean space.
$\rho(\mathbf{A})$ is the spectral radius of $\mathbf{A}$, i.e., the  largest absolute value of its eigenvalues.
$\left[u\right]_{B \times B}$ denotes the $B \times B$ matrix with identical elements~$u$.
$[\mathbf{A}]_{j,k}$ denotes the element at the $j$th row and $k$th column of a matrix $\mathbf{A}$.
$\{v\}_{\mathbb{N}_0}$ denotes the semi-infinite sequence   $\{v_0,v_1,\cdots\}$.

\section{System Model} \label{sec:sys}
We consider a basic system setting wherein a smart sensor periodically samples, pre-estimates and sends its local estimation of a dynamic process to a remote estimator through a wireless link affected by random packet dropouts, as illustrated in Fig.~\ref{retransmit_system_model}.

\subsection{Process Model and Smart Sensor}
The discrete-time linear time-invariant (LTI) model is given as (see e.g., \cite{schenato2008optimal,shi2012optimal,yang2013schedule})
\begin{equation} \label{sys}
\begin{aligned}
\mathbf{x}_{t+1} &= \mathbf{A} \mathbf{x}_t + \mathbf{w}_t,\\
\mathbf{y}_t &= \mathbf{C}\mathbf{x}_t + \mathbf{v}_t,
\end{aligned}
\end{equation}
where 
$\mathbf{x}_t \in \mathbb{R}^n$ is the process state vector, $\mathbf{A} \in \mathbb{R}^{n \times n}$ is the state transition matrix, $\mathbf{y}_t \in \mathbb{R}^m$ is the measurement vector of the smart sensor attached to the process, $\mathbf{C} \in \mathbb{R}^{m \times n}$ is the measurement matrix, $\mathbf{w}_t \in \mathbb{R}^n$ and $\mathbf{v}_t \in \mathbb{R}^m$ are the process and measurement noise vectors, respectively. We assume $\mathbf{w}_t$ and $\mathbf{v}_t$ are independent and are identically distributed (i.i.d.) zero-mean Gaussian processes with corresponding covariance matrices $\mathbf{W}$ and $\mathbf{V}$, respectively. 
In this work, we focus on the stability condition of the remote estimation of the process $\mathbf{x}_t$ in the sense of average remote estimation mean-square error. Note that, if $\rho^2(\mathbf{A}) <1$, then the covariance of $\mathbf{x}_t$ is always bounded, and stability will trivially be satisfied.
Thus, as commonly done in this context, in the sequel we focus on the more interesting case with $\rho^2(\mathbf{A}) \geq 1$ indicating that the plant state grows up exponentially fast.
	
Note that in practice, feedback control of open-loop unstable LTI systems with Gaussian noise always requires sensors with unbounded measurement range, and adaptive zooming-in/zooming-out measurement range have been widely adopted~\cite{Nair}. Beyond the idealized situation of LTI systems, system models with exponentially growing modes are also obtained through linearization at unstable equilibria, see for example the Pendubot system in~\cite{AlexSecurity} and references therein.

Since the sensor's measurements are noisy, a smart sensor is used to estimate the state of the process, $\mathbf{x}_t$. For that purpose  a Kalman filter~\cite{shi2012optimal,yang2013schedule} is used, which gives the minimum estimation MSE, based on the current and previous raw measurements:
\begin{subequations}\label{sub:1}
	\begin{align}
	\mathbf{x}_{t|t-1}^s&=\mathbf{A} \mathbf{x}_{t-1|t-1}^s\\
	\mathbf{P}_{t|t-1}^s&=\mathbf{A} \mathbf{P}_{t-1|t-1}^s \mathbf{A}^{\top}+\mathbf{W}\\
	\mathbf{K}_t&=\mathbf{P}_{t|t-1}^s \mathbf{C}^{\top}(\mathbf{C} \mathbf{P}_{t|t-1}^s \mathbf{C}^{\top}+\mathbf{V})^{-1}\\
	\mathbf{x}_{t|t}^s&=\mathbf{x}_{t|t-1}^s+\mathbf{K}_t(\mathbf{y}_{t}-\mathbf{C} \mathbf{x}_{t|t-1}^s)\\
	\mathbf{P}_{t|t}^s&=(\mathbf{I}-\mathbf{K}_{t} \mathbf{C})\mathbf{P}_{t|t-1}^s
	\end{align}
\end{subequations}
where $\mathbf{I}$ is the $m \times m$ identity matrix, $\mathbf{x}^s_{t|t-1}$ is the prior state estimate, $\mathbf{x}^s_{t|t}$ is the posterior state estimate at time $t$, $\mathbf{K}_t$ is the Kalman gain. The matrices $\mathbf{P}^s_{t|t-1}$ and $\mathbf{P}^s_{t|t}$ represent the prior and posterior error covariance at the sensor at time $t$, respectively. The first two equations above present the prediction steps while the last three equations correspond to the updating steps~\cite{maybeck1979stochastic}.
Note that $\mathbf{x}^s_{t|t}$ is the output of the Kalman filter at time $t$, i.e., the pre-filtered measurement of~$\mathbf{y}_t$, with the estimation error covariance $\mathbf{P}_{t|t}^s$.

\begin{assumption}[\cite{schenato2008optimal,shi2012optimal,shi2012scheduling,wu2018optimal,leong2017sensor,yang2013schedule}]\label{assum}
	\normalfont
	$(\mathbf{A},\mathbf{C})$ is observable and $(\mathbf{A},\sqrt{\mathbf{W}})$ is controllable, i.e., the matrix concatenations $\left[\mathbf{C}^{\top},\mathbf{A}^{\top}\mathbf{C}^{\top},\cdots,(\mathbf{A}^n)^{\top}\mathbf{C}^{\top}\right]$ and $\left[\sqrt{\mathbf{W}},\mathbf{A}\sqrt{\mathbf{W}},\cdots,\mathbf{A}^n\sqrt{\mathbf{W}}\right]$ are of full rank.
\end{assumption}

Using Assumption~\ref{assum}, the local Kalman filter of system \eqref{sys} is stable, i.e.,
the error covariance matrix $\mathbf{P}_{t|t}^s$ converges to a finite matrix~$\bar{\mathbf{P}}_0$ as the time index $t$ goes to infinity~\cite{maybeck1979stochastic}.
In the rest of the paper, we assume that the local Kalman filter operates in the steady state~\cite{schenato2008optimal,shi2012optimal,shi2012scheduling,wu2018optimal,leong2017sensor,yang2013schedule}, i.e., $\mathbf{P}_{t|t}^s = \bar{\mathbf{P}}_0$. 
Further, to simplify   notation, the sensor's estimation $\mathbf{x}_{t|t}^s$ shall be denoted by $\hat{\mathbf{x}}_t^s$.

\subsection{Wireless Channel}   \label{sec:ARQ} 
The main characteristic of the wireless fading channel is that the channel quality is a time-varying random process that changes over time in a correlated manner~\cite{Parastoo,markovchannel1,markovchannel2}.
We consider a finite-state time-homogeneous Markov block-fading channel~\cite{Parastoo}.
It is assumed that the channel power gain $h_t>0$ remains constant during the $t$th time slot but may change slot by slot. We assume that the Markov channel has $M$ states, i.e, $$h_t \in\mathcal{B}\triangleq\{b_1,...,b_M\}.$$
The transition probability  from state $i$ to state $j$ is time-homogeneous and given by 
\begin{equation}
p_{i,j} \triangleq \myprob{ h_{t+1} =b_j \vert h_t =b_i}, \forall i,j \in \mathcal{M}, t\in\mathbb{N}_0,
\end{equation}
where $\mathcal{M}\triangleq \{1,\cdots,M\}$. 
The matrix of channel state transition probability is given as
\begin{equation} \label{M_matrx}
\mathbf{M} \triangleq \begin{bmatrix}
p_{1,1}  & \cdots &	p_{M,1}\\
\vdots  & \ddots &	\vdots\\
p_{1,M}  & \cdots &	p_{M,M}
\end{bmatrix}.
\end{equation}
We assume that all the channel states are \emph{aperiodic} and \emph{positive recurrent}. Thus, the Markov chain induced by $\mathbf{M}$ is \emph{ergodic}~\cite{durrett2019probability}.

We assume that the channel state information is available at both the sensor and the remote estimator, which can be achieved by standard channel estimation and feedback techniques, see e.g.~\cite{channelestimation} and the references therein.
Let $\gamma_t=1$ and $\gamma_t =0$ denote the successful and failed packet detection of the remote estimator during time slot $t$, respectively.
The packet drop probability in channel state $b_i$ is 
\begin{equation}
d_i \triangleq \myprob{\gamma_t = 0 \vert h_t=b_i},  \forall i \in \mathcal{M}, t\in \mathbb{N}_0.
\end{equation}
Note that 
the transmission is always perfect if $d_i=0,\forall i \in \mathcal{M}$, while no information-carrying packet is delivered to the remote estimator if $d_i=1,\forall i \in \mathcal{M}$. 
We define the packet drop probability matrix as
\begin{equation}\label{D_matrx}
\mathbf{D} \triangleq \text{diag}\{d_1,d_2,\cdots,d_M\}.
\end{equation}

\begin{example}
\normalfont
Suppose that the Markov channel has only two states, where the channel power gains, i.e., the effective signal-to-noise ratios (SNRs), are $b_1 = 300$ and $b_2=250$. Assume the estimate-carrying packet has $\zeta=200$ symbols and each symbol carries $R=8$ bits information. The minimum achievable packet drop rate is~\cite{Polyanskiy}
\begin{equation} \label{Yury_e}
\varepsilon \approx Q\left(\sqrt{\frac{\zeta}{\nu}}\left(C-R\right)\right),
\end{equation}
where $Q(x) = \frac{1}{2 \pi} \int_{x}^{\infty} \exp\left(-\frac{u^2}{2}\right) \mathrm{d}u$, 
$$
C = \log_2\left(1+h\right),\ 
\nu =h \frac{2+h}{\left(1+h\right)^2} \left(\log_2 e\right)^2,
$$
and $h$ is the SNR and $e$ is the Euler's number.
Taking $b_1$, $b_2$, $\zeta$ and $R$ into \eqref{Yury_e}, the packet drop probabilities are $d_1= 0.0039$ and $d_2=0.2584$.
\end{example}

In practice, the sensor can send a known sequence of symbols (called a pilot or channel training sequence) to the receiver, which can then estimate the channel power gain~\cite{tse2005fundamentals}. For the packet error probability at a certain channel condition, accurate value can be obtained by Monte Carlo simulation, i.e., sending a large sequence of packets to the receiver and calculate the ratio of failed packets.

\subsection{Remote Estimation and Stability Criteria}\label{sec:estimation}
The smart sensor sends its local estimate $\hat{\mathbf{x}}^s_t$ to the remote estimator at every time slot. Each packet transmission has a unit delay that is equal to the sampling period of the system.
For the considered fading model, packets may or may not arrive at the receiver due to the random packet dropouts. 
To account for packet transmission delays and the failures, the remote estimation of the current system states is based on the previously detected information packet.
The optimal remote estimator in the sense of minimum mean-square error (MMSE) can be obtained as \cite{schenato2008optimal,shi2012scheduling}
\begin{equation}\label{eq:estimation}
\hat{\mathbf{x}}_t = \begin{cases}
\mathbf{A} \hat{\mathbf{x}}_{t-1},& \gamma_{t-1}=0,\\
\mathbf{A} \hat{\mathbf{x}}^s_{t-1},& \gamma_{t-1}=1.
\end{cases}
\end{equation}

Assuming a packet was successfully received at time $t'$, and the following transmission consecutively failed for $\delta\geq 1$ times before the current time $t$, i.e., $t=t'+\delta + 1$, from \eqref{eq:estimation}, it can be obtained that $\hat{\mathbf{x}}_{t'+1}=\mathbf{A} \hat{\mathbf{x}}^s_{t'}$, $\hat{\mathbf{x}}_{t'+2}=\mathbf{A}^2 \hat{\mathbf{x}}^s_{t'}$, and $\hat{\mathbf{x}}_{t} = \hat{\mathbf{x}}_{t'+ \delta +1}=\mathbf{A}^{\delta +1} \hat{\mathbf{x}}^s_{t'}$.
Thus, \eqref{eq:estimation} can be uniformly written as
\begin{equation}\label{general_estimater}
\hat{\mathbf{x}}_t = \mathbf{A}^{\delta_t+1} \hat{\mathbf{x}}^s_{t-(\delta_t+1)},
\end{equation}
where $\delta_t \in \mathbb{N}_0$ is the number of consecutive packet dropouts before time slot $t$. In other words, $(\delta_t+1)$ can be treated as the age-of-information (AoI) of the remote estimator in time slot $t$~\cite{kaul2012real}.

Then, the estimation error covariance is given as
\begin{align} \label{covariance1}
\mathbf{P}_t 
&\triangleq  \mathsf{E}\left[(\hat{\mathbf{x}}_t-\mathbf{x}_t)(\hat{\mathbf{x}}_t-\mathbf{x}_t)^{\top}\right]\\
&=v^{(\delta_{t}+1)}(\bar{\mathbf{P}}_0) \label{general_form}
\end{align}
where \eqref{general_form} is obtained by substituting \eqref{general_estimater} and \eqref{sys}  into \eqref{covariance1} and with:
\begin{equation}\label{eq:v}
v(\mathbf{X})\triangleq \mathbf{AXA}^{\top}+\mathbf{W}
\end{equation}
$$v^{1}(\cdot) \triangleq v(\cdot), \quad v^{m+1}(\cdot)  \triangleq v (v^{m}(\cdot)),\quad m\geq 1. $$ 

Thus, the quality of the remote estimation error in time slot $t$ can be quantified via $\text{Tr}\left(\mathbf{P}_t \right)$. 
We introduce the following function
\begin{equation}\label{eq:c}
c(i)\triangleq \text{Tr}\left(v^i(\bar{\mathbf{P}}_0)\right), \forall i\in \mathbb{N}.
\end{equation}
From \eqref{general_form}, we can write
\begin{equation} \label{trace}
\text{Tr}\left(\mathbf{P}_t \right) \triangleq c(\delta_t+1).
\end{equation}
Since  $\mathbf{P}_t $ is a  countable stochastic process taking value from a countable infinity set $$\{v^1(\bar{\mathbf{P}}_0),v^2(\bar{\mathbf{P}}_0),\dots\}$$
it will grow during periods of consecutive packet dropouts when $\rho(\mathbf{A})\geq 1$. Since periods of consecutive packet dropouts have unbounded support, at best one can hope for some type of stochastic stability. In the present work, our focus is on the mean-square stability.

\begin{definition}[Mean-Square Stability]
	\normalfont
	The remote estimation system is mean-square stable if and only if
	the average estimation MSE $J$ is bounded, where
\begin{equation}\label{longterm}
J\triangleq\limsup_{T\to\infty}\frac{1}{T}\sum_{t=1}^{T} \mathsf{E}\left[\text{Tr}\left(\mathbf{P}_t\right)\right],
\end{equation}
and $\limsup_{T\rightarrow \infty}$ is the limit superior operator.	
		
\end{definition}
Note that establishing necessary and sufficient stability conditions is non-trivial as we consider correlated fading-channel model in the remote estimation system which induces a countable (and unbounded) Markov chain in the analysis. Some of the existing works adopt stochastic Lyapunov functions to elucidate such situations (see e.g.,~\cite{Quevedo}). These however,  merely lead to sufficient conditions.

\section{Main Results} \label{sec:main_results}
In this section, we present and discuss the main results of the paper, which will be proved in Section~\ref{sec:main}.
\subsection{The Necessary and Sufficient Stability Condition}
\begin{theorem}\label{theorem:main}
	\normalfont
	Let Assumption~\ref{assum} hold. 		
	The remote estimation system described by~\eqref{sys}, \eqref{sub:1} and \eqref{general_estimater} is mean-square stable over the Markov channel defined by~\eqref{M_matrx} and \eqref{D_matrx} if and only if the following condition holds:
	\begin{equation} \label{stability_condition}
	\rho^2(\mathbf{A}) \rho\left(\mathbf{D} \mathbf{M}\right) <1.
	\end{equation}
\end{theorem}
Theorem~\ref{theorem:main} shows that the stability condition depends on the system matrix $\mathbf{A}$, the packet drop probability matrix $\mathbf{D}$ and the matrix of the channel state transitions $\mathbf{M}$. It is important to note that the necessary and sufficient condition is determined by both the spectral radiuses of $\mathbf{A}$ and the product of two matrices $\mathbf{D}$ and $\mathbf{M}$. Since $\rho(\mathbf{A})$ measures how fast the dynamic process varies, $\rho\left(\mathbf{D} \mathbf{M}\right)$ can be treated as an effective measurement of the Markov channel quality.

\begin{remark}
	In~\cite[Corollary 1]{Quevedo}, a sufficient condition in terms of exponential stability of a conventional-sensor-based remote estimation system over Markov channel is obtained as
	\begin{equation}\label{eq:Daniel}
	\tilde{\rho}^2(\mathbf{A})\max_{i\in\mathcal{M}}\left\lbrace \sum_{j=1}^{M} p_{ij} d_j\right\rbrace <1,
	\end{equation}
	where $\tilde{\rho}(\mathbf{A})$ is the largest singular value of $\mathbf{A}$.
	Using Perron–Frobenius theorem~\cite{Perron}, we have $\max_{i\in\mathcal{M}}\left\lbrace \sum_{j=1}^{M} p_{ij} d_j\right\rbrace> \rho(\mathbf{MD}) = \rho(\mathbf{DM})$.
	In addition, due to the fact that the largest singular value is no smaller than the spectral radius, i.e., $\tilde{\rho}(\mathbf{A})\geq\rho(\mathbf{A})$, it can be proved that the sufficient condition~\eqref{eq:Daniel} is more restrictive than~\eqref{stability_condition}.	
\end{remark}

\begin{corollary}[Special Case I] \label{cor:1}
	\normalfont
	Consider the same assumption and system model in Theorem~\ref{theorem:main}. For the special case of i.i.d. packet dropout channel with packet dropout probability $d$, the remote estimation system is mean-square stable if and only if the following condition holds:
	\begin{equation} 
	\rho^2(\mathbf{A}) d <1.
	\end{equation}
\end{corollary}
\begin{remark}
	For the special case of i.i.d. packet dropouts, our stability condition obtained from Theorem~\ref{theorem:main} is identical to the conventional result in~\cite{schenato2008optimal}.
\end{remark}

\begin{corollary}[Special Case II]\label{cor:2}
	\normalfont
	Consider the same assumption and system model in Theorem~\ref{theorem:main}. For the special case of Markovian packet dropout channel with packet dropout probability $d_1=0$ and $d_2=1$ and the channel state transition matrix $$\mathbf{M}=\begin{bmatrix}
	p_{11}&p_{12}\\
	p_{21}&p_{22}	
	\end{bmatrix},$$ the remote estimation system is mean-square stable if and only if the following condition holds:
	\begin{equation} \label{eq:specialcase2}
	\rho^2(\mathbf{A}) p_{22} <1.
	\end{equation}
\end{corollary}
\begin{remark}
	For the Markovian on-off channel in Corollary~\ref{cor:2}, it is interesting to see that the stability condition only depends on one element of the $2$-by-$2$ matrix $\mathbf{M}$, which is the state transition probability from the bad state to the bad state.
	
	We would like to compare our result with the one obtained in~\cite{XieXie}, which considered a conventional sensor scenario.
	In~\cite[Theorem 2]{XieXie}, a necessary stability condition is obtained as
	\begin{equation}
	\rho^2(\mathbf{A}) \min\{p_{22},(1-p_{12})\}<1,
	\end{equation}
	which is less restrictive than our current result~\eqref{eq:specialcase2}.
\end{remark}

\subsection{Upper and Lower Bounds of the Estimation Error Function}
A pair of asymptotic upper and lower bounds of the estimation error function are given below.
\begin{proposition}[Asymptotic upper bound of the estimation-error function]\label{lem:cost_up}
	\normalfont
	For any $\epsilon>0$, there exists $N>0$ and $\kappa >0$ such that $$c(i) < \kappa \left(\rho^2(\mathbf{A})+\epsilon \right)^i, \forall i>N.$$
\end{proposition}

\begin{proposition}[Asymptotic lower bound of the estimation-error function]\label{lem:cost_low}
	\normalfont
	There exists a constant $N>0$ and $\eta>0$ such that $c(i) \geq \eta (\rho(\mathbf{A}))^{2i}, \forall i>N$.
\end{proposition}
Propositions~\ref{lem:cost_up} and \ref{lem:cost_low} show that when a large number of consecutive packet dropouts occur, i.e., $i\gg 1$,  the remote estimation error is upper and lower bounded by exponential functions in terms of $i$.

\begin{remark}
It can be observed that the estimation-error function $c(i)$ grows as exponentially fast as $\rho^{2i}(\mathbf{A})$.
\end{remark}

\section{Analysis of the Average Estimation MSE}\label{sec:proof}
In this section, we first investigate an estimation-cycle based performance analysis approach of the remote state estimation, and then develop new element-wise bounds of matrix powers. The results and technical lemmas obtained in this section will be used for the proofs of the main results of the paper. 

As it is clear that Theorem 1 holds for the special cases with $\mathbf{D} = \mathbf{0}$ or $\mathbf{I}$, in the following we only focus on the cases with $\mathbf{D} \neq \mathbf{0}$ nor $\mathbf{I}$.

\subsection{Stochastic Estimation-Cycle Based Analysis}
Before analyzing the long-term average MSE of the remote estimation system and derive the stability condition, we need to introduce and analyze estimation cycle.
To be specific, the $k$th estimation cycle starts after the $k$th successful transmission and ends at the $(k+1)$th successful transmission as illustrated in Fig.~\ref{fig:process}. In other words, the estimation process is divided by the estimation cycles.

\begin{figure}[t]
	\centering
	\includegraphics[scale=1.0]{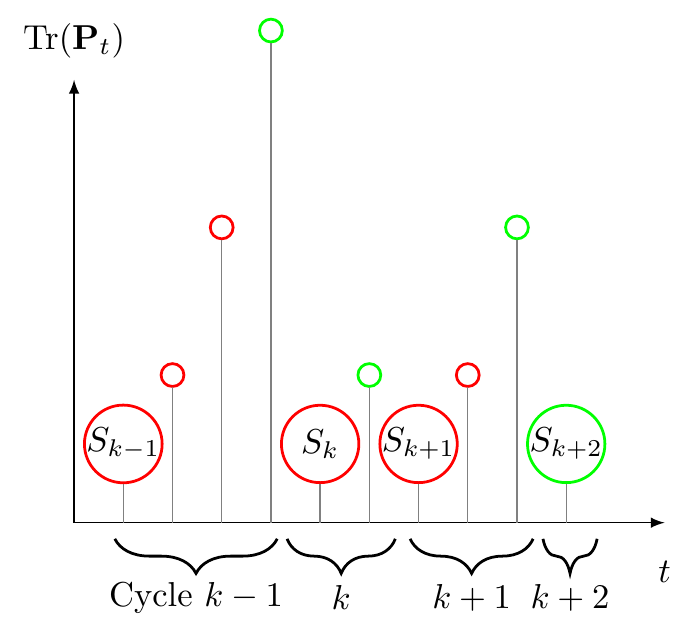}
	\caption{Illustration of estimation cycles, where red and green circles denote failed and successful transmissions, respectively, and big circles denote the beginning of estimation cycles.}
	\label{fig:process}
\end{figure}

The channel state at the beginning of estimation cycle $k$, i.e., a post-success channel state, is denoted by $S_k\in \mathcal{B}' \subset \mathcal{B}$, where
\begin{equation}\label{eq:B'}
\mathcal{B}' \triangleq \{b_j: \max_{i\in\mathcal{M}} (1-d_i) p_{i,j}  >0, \forall j\in \mathcal{M} \}
\end{equation}
is the set of post-success channel states and the cardinality of $\mathcal{B}'$ is $M' \leq M$.
Without loss of generality, we assume that $\mathcal{B}'$ contains the first $M'$ elements of $\mathcal{B}$. In other words, none of the last $(M-M')$ elements of $\mathcal{B}$ can be a post-success channel state, while the others can.
\begin{example}
\normalfont	
	Consider a two-state Markov channel with $$\mathbf{M} =\begin{bmatrix}
	0 & 1\\
	1 & 0\\
	\end{bmatrix},$$ $d_1=1$ and $d_2=0$. It is clear that the channel states deterministically switches between the two states, and the transmission can be successful only in channel state 2. Thus, channel state 1 is the only post-success state, i.e., $\mathcal{B'}=\{b_1\} \subset \mathcal{B}=\{b_1,b_2\}$. 
\end{example}

Then, we have the following property of $S_k$.
\begin{lemma}\label{lem:G}
\normalfont	
$\{S\}_{\mathbb{N}_0}$ is a time-homogeneous ergodic Markov chain with $M'$ irreducible states of $\mathcal{B}'$. 
The state transition matrix of $\{S\}_{\mathbb{N}_0}$ is $\mathbf{G}'$, which is the $M'$-by-$M'$ matrix taken from the top-left corner of $\mathbf{G}$, where
\begin{equation}\label{eq:G}
\mathbf{G}=\sum_{j=0}^{\infty} (\mathbf{DM})^j \mathbf{(I-D)M},
\end{equation}
and the last $(M-M')$ columns of $\mathbf{G}$ are all zeros.
The stationary distribution of $\{S\}_{\mathbb{N}_0}$ is $\boldsymbol{\beta} \triangleq [\beta_1,\cdots,\beta_{M'}]^\top$, which is the unique null-space vector of $(\mathbf{I-G'})^\top$ and $\beta_i >0,\forall i\in \mathcal{M}'$, where $\mathcal{M}'\triangleq \{1,2,\cdots, M'\}$.
\end{lemma}
\begin{proof}
See Appendix A.
\end{proof}

\begin{remark}
Our analysis investigates the sequence of successful reception instances. 
This has also been considered in~\cite{anytime}, where the instances of successful reception are return times of a Markov chain. Different to~\cite{anytime}, we focus on the channel states right after these instances, which form an ergodic Markov chain. Our approach will shed lights on the future work of the analysis of closed-loop control systems over Markov channels.
\end{remark}

Let $T_k$ denote the sum number of transmissions in the $k$th estimation cycle.
The sum MSE in the $k$th estimation cycle, say $C_k$, is given as
\begin{equation}\label{g_fun}
C_k = g(T_k) \triangleq \sum_{j=1}^{T_k} c(j).
\end{equation}
From \eqref{longterm} it directly follows that  the average estimation MSE  can be rewritten as
\begin{align}\label{J}
J &=\limsup\limits_{K \rightarrow \infty} \frac{C_1+C_2+\cdots+C_K}{T_1+T_2+\cdots+T_K}.
\end{align}
Since $C_k$ is determined by $T_k$, and the distribution of $T_k$ depends on $S_k$ and the distribution of $S_k$ is time-invariant, the  unconditional distributions of $T_k$ and of $C_k$ are also time-invariant.
We thus drop the time indexes of $T_k$, $C_k$ and $S_k$.
Then, the time average of $\{\cdots, T_k,T_{k+1},\cdots\}$ and $\{\cdots, C_k,C_{k+1},\cdots\}$  can be translated to the following ensemble averages as 
\begin{equation}\label{expect_T}
\myexpect{T}=\lim\limits_{K\rightarrow \infty }\frac{1}{K} \sum_{k=1}^{K} T_k =  \sum_{m=1}^{M} \beta_m \myexpect{T \vert S=b_m},
\end{equation}
and
\begin{equation}\label{expect_C}
\myexpect{C} =\lim\limits_{K\rightarrow \infty }\frac{1}{K} \sum_{k=1}^{K} C_k =  \sum_{m=1}^{M} \beta_m \myexpect{C \vert S=b_m},
\end{equation}
where $\beta_m$ is defined in Lemma~\ref{lem:G} for $m\in\{1,\cdots,M'\}$ and $\beta_m=0$ when $m>M'$.

From the definition of estimation cycle and the property of channel state transition, the conditional probability of the length of an estimation cycle is obtained as
\begin{equation}\label{prob}
\myprob{T=i \vert S=b_m}= \sum_{k=1}^{M} \left[\mathbf{(DM)}^{i-1} \mathbf{(I-D)M}\right]_{m,k}.
\end{equation}
If we now replace \eqref{prob} into \eqref{expect_T} and into \eqref{expect_C}, then   after some algebraic manipulations, one can obtain:
\begin{align}\label{T}
\myexpect{T}&=\sum_{j=1}^{M}\sum_{k=1}^{M} \sum_{i=1}^{\infty} i \left[\mathbf{\Xi}(i)\right]_{j,k},\\ \label{C}
\myexpect{C}&=\sum_{j=1}^{M}\sum_{k=1}^{M} \sum_{i=1}^{\infty} g(i) \left[\mathbf{\Xi}(i)\right]_{j,k},
\end{align}
where $$\mathbf{\Xi}(i) = \text{diag}\{\beta_1,\cdots,\beta_{M}\}\mathbf{(DM)}^{i-1} \mathbf{(I-D)M}.$$
Taking \eqref{T} and \eqref{C} into \eqref{J}, we have
\begin{align}\label{J2}
J =\limsup\limits_{K \rightarrow \infty} \frac{\frac{1}{K} (C_1+C_2+\cdots+C_K)}{\frac{1}{K}(T_1+T_2+\cdots+T_K)}
= \frac{\mathsf{E}\left[C\right]}{\mathsf{E}\left[T\right]}.
\end{align}

Therefore, it turns out that the average estimation MSE $J$ depends on the estimation error function $c(i)$ and the function $(\mathbf{DM})^i\mathbf{(I-D)M}$, both of which involve matrix powers.
In what follows, we will introduce and prove some technical lemmas about the element-wise upper and lower bounds of matrix powers, which are the key steps for analyzing the sufficient and necessary stability conditions of the remote estimation system.

\subsection{Element-Wise Bounds of Matrix Powers}\label{sec:ele}
We give an element-wise upper bound of matrix powers as below.
\begin{lemma}[Element-wise upper bound of matrix power]\label{lem:matrix_up}
	\normalfont
Consider a $z$-by-$z$ matrix $\mathbf{Z}$ with $A$ different eigenvalues $\{\lambda_1,\lambda_2,\cdots,\lambda_A\}$, where $1 \leq A\leq z$, and define $\mathcal{Z}\triangleq \{1,\cdots,z\}$. Then, for any $\epsilon>0$, there exist $N>0$ and $\kappa >0$ such that 
$$\vert [\mathbf{Z}^i]_{j,k} \vert^2 < \kappa \left(\rho(\mathbf{Z})+\epsilon \right)^{2i}, \forall j,k\in \mathcal{Z}, \forall i>N.$$
\end{lemma}
\begin{proof}
See Appendix B.
\end{proof}

\begin{definition}[Asymptotically and periodically lower bounded] \label{def:aplowerbound}
	\normalfont
	A function $r(k)$ is asymptotically and periodically lower bounded by $\underline{r}(k)$ with a period $l\in\mathbb{N}$ if there exists $N\in\mathbb{N}$ such that $$\max\{r(i),r(i+1),\cdots,r(i+l-1)\}\geq \underline{r}(i),\forall i \geq N.$$
\end{definition}
Thus, if $r(k)$ is asymptotically and periodically lower bounded by $\underline{r}(k)$ with a period $l$, then the sum of the function $r(k)$ with a consecutive of $l$ samples is lower bounded by $\underline{r}(k)$.
When a direct lower bound of the function $r(k)$ is intractable or very loose, we can resort to finding a periodical lower bound $\underline{r}(k)$, which might introduce a tight lower bound of the average of $r(k)$ per $l$~samples, i.e., $\underline{r}(k)/l$.
It is clear that the periodical lower bound is tighter if the period $l$ is smaller. In Lemma~\ref{lem:matrix_low_1}, we will show how to determine the period of a specific problem in details.
Definition~\ref{def:aplowerbound} will be used to capture the lower bound of the average sum MSE in \eqref{g_fun} for analyzing the necessary stability condition.

\begin{definition} [Asymptotically lower bounded]
	\normalfont	
A function $r(k)$ is asymptotically lower bounded by $\underline{r}(k)$ if it is asymptotically and periodically lower bounded by $\underline{r}(k)$ with period $1$.
\end{definition}

Given the preceding definitions, we can obtain an element-wise lower bound of matrix powers as below.
\begin{lemma}[Element-wise lower bound of matrix power]\label{lem:matrix_low_1}
	\normalfont
\ 
\\
	(i) Consider a $z$-by-$z$ matrix $\mathbf{Z}$.
	Then there exist $\eta>0$ and $j,k\in\mathcal{Z}$ such that $\left\lvert [\mathbf{Z}^i]_{j,k} \right\rvert^2$ is asymptotically and periodically lower bounded by $\eta (\rho(\mathbf{Z}))^{2i}$. The period is a positive integer no larger than the number of eigenvalues of $\mathbf{Z}$ with the same maximum magnitude.
\\	
	(ii)
Consider a pair of $z$-by-$z$ matrices $\mathbf{Z}$ and $\mathbf{Q}$ with the assumptions that $\mathbf{Q}$ is symmetric positive semidefinite and $(\mathbf{Z},\sqrt{\mathbf{Q}})$ is controllable.
Then there exist $\eta>0$ and $j,k\in\mathcal{Z}$ such that $\left\lvert [\mathbf{Z}^i\sqrt{\mathbf{Q}}]_{j,k} \right\rvert^2$ is asymptotically and periodically lower bounded by $\eta (\rho(\mathbf{Z}))^{2i}$. 
The period has the same property as in (i).
\end{lemma}
\begin{proof}
See Appendix B.
\end{proof}

\section{Proof of the Main Results}\label{sec:main}
In this section, we prove Propositions~\ref{lem:cost_up} and \ref{lem:cost_low}, and Theorem~\ref{theorem:main}. 
\subsection{Proof of Proposition~\ref{lem:cost_up}}
	Taking \eqref{eq:v} into \eqref{eq:c}, we have
	\begin{equation}\label{cost}
	\begin{aligned}
	c(i)=\!\text{Tr}\left(\!\!\mathbf{A}^i\sqrt{\mathbf{\bar{P}}_0}(\mathbf{A}^i\sqrt{\mathbf{\bar{P}}_0})^\top\!\right)\!\!+\!\!\sum_{m=0}^{i-1}\!\text{Tr}\left(\!\mathbf{A}^m\mathbf{\sqrt{W}}(\mathbf{A}^m\mathbf{\sqrt{W}})^\top\!\right)\!.
	\end{aligned}
	\end{equation}
	From \eqref{cost} and Lemma~\ref{lem:matrix_up}, for any $\epsilon>0$, there exists $\kappa,\kappa',N>0$ such that for all $i>N$ we have
\begin{equation}
	\begin{aligned}
	&c(i)\leq\!\! n^2 \left(\!\max_{j,k \in \mathcal{N}} \left(\!\!\left[\mathbf{A}^i\sqrt{\mathbf{\bar{P}}_0}\right]_{j,k}\right)^2 \!\!\!\!+ \!\!\sum_{m=0}^{i-1} \max_{j,k \in \mathcal{N}} \left(\!\!\left[\mathbf{A}^m\sqrt{\mathbf{W}}\right]_{j,k}\!\right)^{\!2}\!\! \right)\\
	&\leq n^2\left(\kappa (\rho(\mathbf{A})+\epsilon)^{2i} +\sum_{m=N+1}^{i}\kappa' (\rho(\mathbf{A})+\epsilon)^{2m}  \right.\\
	&\left. \qquad\qquad+\sum_{m=0}^{N} \max_{j,k \in \mathcal{N}} \left(\left[\mathbf{A}^m\sqrt{\mathbf{W}}\right]_{j,k}\right)^2\right)\\
	&\leq n^2\left((i-N+1)\max\{\kappa,\kappa'\} (\rho(\mathbf{A})+\epsilon)^{2i} \right.\\
	&\left. \qquad\qquad+ \sum_{m=0}^{N} \max_{j,k \in \mathcal{N}} \left(\left[\mathbf{A}^m\sqrt{\mathbf{W}}\right]_{j,k}\right)^2\right),
	\end{aligned}
\end{equation}
	where $\mathcal{N}\triangleq\{1,\cdots,n\}$. Recall that $\mathbf{A}$ is an $n$-by-$n$ matrix.
	Thus, for any $\epsilon'>\epsilon$, we can find $N'>N$ and $\kappa''>0$ such that $c(i)\leq \kappa''(\rho(\mathbf{A})+\epsilon')^{2i},\forall i>N'$.
	This completes the proof of Proposition~\ref{lem:cost_up}.

\subsection{Proof of Proposition~\ref{lem:cost_low}}
From Lemma~\ref{lem:matrix_low_1}(ii), we note that
there exists $\eta>0$ and $j,k\in\mathcal{N}$ such that $\left\lvert [\mathbf{A}^i\mathbf{W}]_{j,k} \right\rvert^2$ is asymptotically and periodically lower bounded by $\eta (\rho(\mathbf{A}))^{2i}$ with the period $l$, which is no larger than the dimension of the matrix $\mathbf{A}$.
Then, from \eqref{cost}, when $i$ is sufficiently large, we have
\begin{align}
c(i)&\geq \sum_{m=i-l}^{i-1}\text{Tr}\left(\mathbf{A}^m\mathbf{\sqrt{W}}(\mathbf{A}^m\mathbf{\sqrt{W}})^\top\right)\\
&\geq \sum_{m=i-l}^{i-1} \left\lvert [\mathbf{A}^m\mathbf{W}]_{j,k} \right\rvert^2\\
&\geq \eta (\rho(\mathbf{A}))^{2(i-l)} =\eta \left(\rho(\mathbf{A})\right)^{-2l} \left(\rho(\mathbf{A})\right)^{2i}. 
\end{align}
This completes the proof of Proposition~\ref{lem:cost_low}.

\subsection{Proof of Theorem~\ref{theorem:main}}
Besides the upper and lower bounds of the estimation error function, to obtain the necessary and sufficient stability condition, we need some additional properties of $(\mathbf{DM})^i$ and $(\mathbf{DM})^i(\mathbf{(I-D)M})$.
\begin{lemma}[Property of matrix $(\mathbf{DM})^i$]\label{lem:matrix_low_2'}
	\normalfont
	Consider the stochastic matrix $\mathbf{M}$ and the diagonal matrix $\mathbf{D}$ defined in \eqref{M_matrx} and \eqref{D_matrx}, respectively.
	Let $\mathcal{J}_0 \triangleq \{j|d_j=0,j\in \mathcal{M}\}$ and $\mathcal{\bar{J}}_0 \triangleq \mathcal{M}\backslash \mathcal{\bar{J}}_0 = \{j|d_j\neq0,j\in \mathcal{M}\} \neq \emptyset$.
	
	(i) If $\mathcal{J}_0 = \emptyset$, there exists $\eta>0$ such that $ \left[(\mathbf{DM})^i\right]_{j,k} $ is asymptotically and periodically lower bounded by $\eta \rho^i(\mathbf{DM}),\forall j,k \in \mathcal{M}$.
	
	(ii) If $\mathcal{J}_0 \neq \emptyset$, there exists $\eta>0$, $j\in \mathcal{\bar{J}}_0$ and $k\in \mathcal{J}_0$ such that $ \left[(\mathbf{DM})^i\right]_{j,k}$ is asymptotically and periodically lower bounded by $\eta \rho^i(\mathbf{DM})$.
	
	(iii) $\rho(\mathbf{DM})<1$.
	
\end{lemma}
\begin{proof}	
See Appendix C.
\end{proof}

\begin{lemma}[Property of matrix $(\mathbf{DM})^i(\mathbf{I-D})\mathbf{M}$]\label{lem:matrix_low_2}
	\normalfont
	Given the stochastic matrix $\mathbf{M}$ and the diagonal matrix $\mathbf{D}$ defined in \eqref{M_matrx} and \eqref{D_matrx}, respectively, there exist $\eta>0$ and $j,k\in\mathcal{M}$ such that $ \left[(\mathbf{DM})^i\mathbf{(I-D)M}\right]_{j,k} $ is asymptotically and periodically lower bounded by $\eta \rho^i(\mathbf{DM})$.
\end{lemma}
\begin{proof}
See Appendix C.
\end{proof}

By using the properties in Lemmas~\ref{lem:matrix_low_2'} and \ref{lem:matrix_low_2},  Theorem~\ref{theorem:main} can be proved as in Appendix D.

\section{Numerical Results} \label{sec:num}
In this section, we illustrate and compare the stability regions of the remote estimation system obtained using Theorem~\ref{theorem:main} of the current paper and based on Corollary 1 of our previous work~\cite{Quevedo}. We also present simulated results of the average estimation MSE in \eqref{J} based on the average of $10^5$ time steps. 

We consider an example involving the Pendubot, a two-link
planar robot~\cite{pendubot1}. A linearized continuous time
model for balancing the Pendubot in the upright position can be found in~\cite{pendubot2}. With a sampling time of $15$ ms, we can then obtain the following discrete time model~\cite{AlexSecurity}:
\begin{subequations}\label{eq:pendubot}
	\begin{align} 
&\mathbf{A} = \begin{bmatrix}
1.0058 &0.0150 &-0.0016 &0.0000\\
0.7808 &1.0058 &-0.2105 &-0.0016\\
-0.0060 &0.0000 &1.0077 &0.0150\\
-0.7962 &-0.0060 &1.0294 &1.0077
\end{bmatrix},\\
&\mathbf{C} = \begin{bmatrix}
1 & 0 & 0 & 0\\
0 & 0 & 1 & 0\\
\end{bmatrix}, \\
&\mathbf{W}=\mathbf{u} \mathbf{u}^\top, \mathbf{u}=\begin{bmatrix}
0.003 & 1.0000 &-0.005 &-2.150\\
\end{bmatrix}^\top,\\
&\mathbf{V}  = 0.001 \times \mathbf{I}.
\end{align}
\end{subequations}
Thus, $\rho(\mathbf{A}) = 1.15$ and $\tilde{\rho}(\mathbf{A}) = 2$. 
Unless otherwise
stated, we consider a two-state Markov channel model characterised by the transition matrix $\mathbf{M} = \begin{bmatrix}
0.1 & 0.9\\0.5 & 0.5
\end{bmatrix}$ and conditional dropout probabilities $d_1 = 0.8$, and $d_2 = 0.1$.

Fig.~\ref{fig:4plots} shows the stability regions (in the dropout probability plane) for different $\mathbf{A}$ and $\mathbf{M}$. In this figure,  the solid and dashed line bounded regions are obtained from Theorem~\ref{theorem:main} and~\cite[Corollary 1]{Quevedo}, respectively.
Specifically, we set $\mathbf{A} = \mathbf{A}_1 = \begin{bmatrix}
1.129 &0.0150 &-0.0016 &0.0000\\
0.7808 &1.0058 &-0.2105 &-0.0016\\
-0.0060 &0.0000 &1.0077 &0.0150\\
-0.7962 &-0.0060 &1.0294 &1.0077
\end{bmatrix}$ in case (b), 
where 
 $\rho(\mathbf{A}_1) = 1.2$ and $\tilde{\rho}(\mathbf{A}_1) = 2$, and set $\mathbf{M} = \mathbf{M}_1 =\begin{bmatrix}
0.1 & 0.9\\0.9 & 0.1
\end{bmatrix}$ and $ \mathbf{M}_2 =\begin{bmatrix}
0.9 & 0.1\\0.1 & 0.9
\end{bmatrix}$ in cases (c) and (d), respectively.
From cases (a)-(d), it is clear that our current necessary and sufficient stability region is much larger than the sufficient stability region established in~\cite{Quevedo}.
Also, it can be observed that the necessary and sufficient stability region is convex in case (d) and is non-convex in cases (a)-(c).
Note that the convexity of a stability region is important in practice. For example, if one has tested a set of communications parameter vectors that can stabilize the remote estimation system and the stability region is proved to be convex, any parameter vector that belong to the convex hull of the set can stabilize the system as well.
Comparing (b) with (a), it is clear to see a smaller stability region as the system in (b) is more unstable than (a).
Comparing (c) with (d), it is interesting to see that if the Markov channel has a longer memory, i.e., it has a higher chance to stay in a poor channel condition, then the remote estimation system has a smaller stability region.

\begin{figure}[t]
	\centering\includegraphics[scale=0.6]{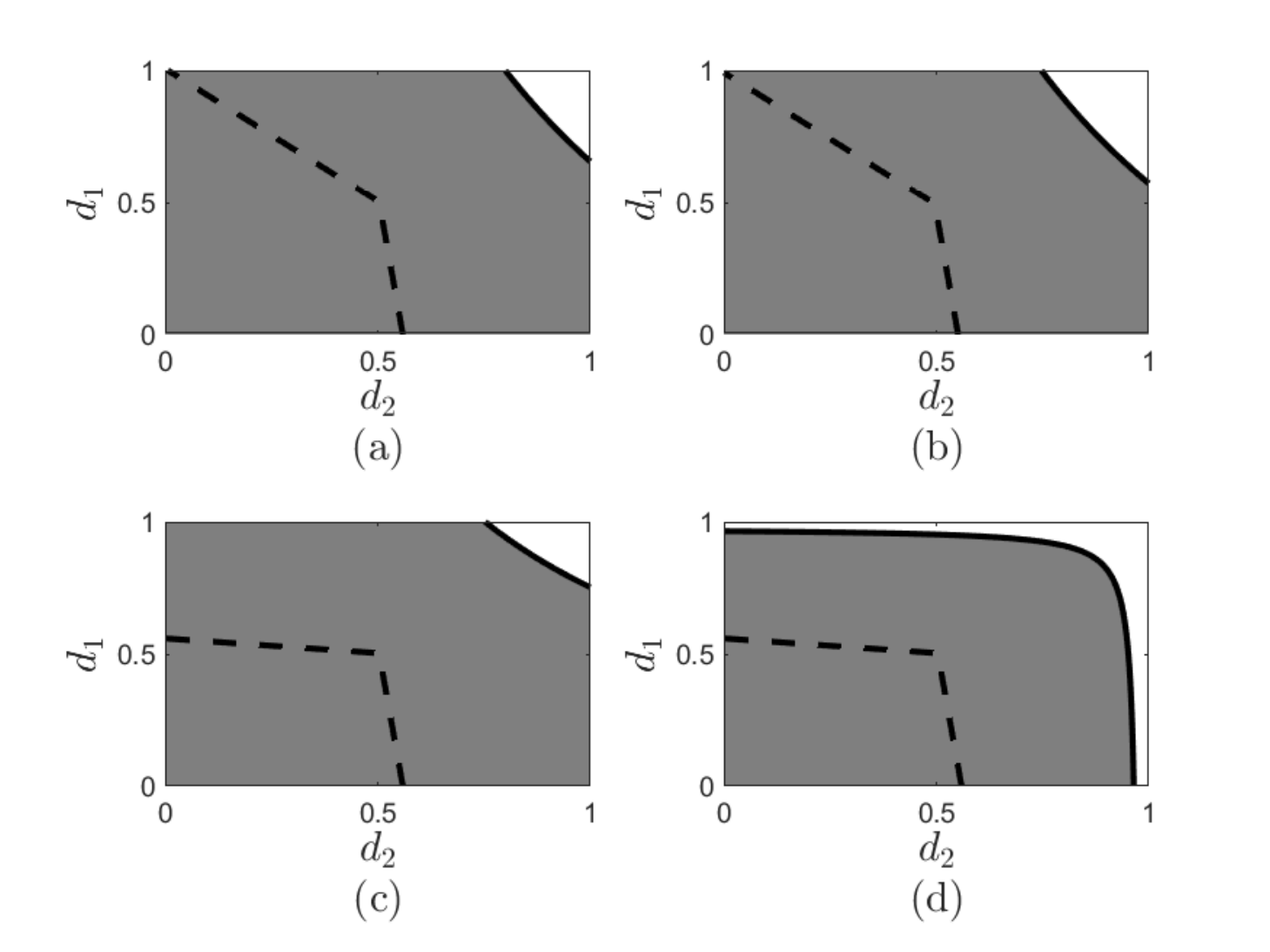}
	\caption{The necessary and sufficient stability region of Theorem~\ref{theorem:main} (i.e., the solid line bounded area) and the sufficient stability region~\cite{Quevedo} (i.e., the dashed line bounded area).}
	\label{fig:4plots}
\end{figure}

Fig.~\ref{fig:4process} shows the original unstable process $\mathbf{x}_t$ of the system~\eqref{eq:pendubot}, and the remote estimation $\mathbf{\hat{x}}_t$. We see that the estimator tracks the unstable process well, and the relative estimation error, i.e., $|\mathbf{x}_t - \mathbf{\hat{x}}_t|/ |\mathbf{x}_t|$ decreases with time.
\begin{figure}[t]
	\centering\includegraphics[scale=0.6]{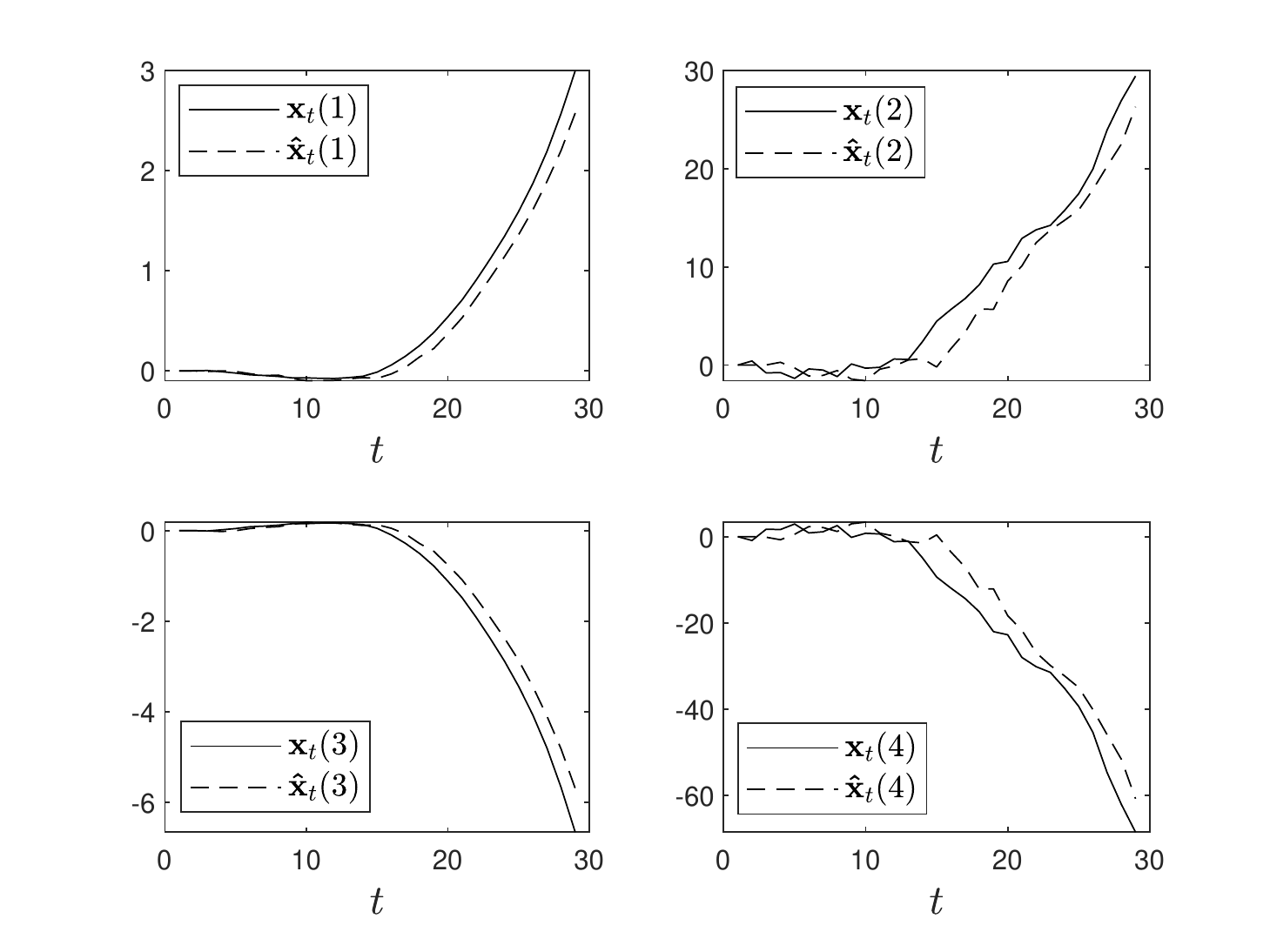}
	\caption{The original process $\mathbf{x}_t \triangleq [\mathbf{x}_t(1),\mathbf{x}_t(2),\mathbf{x}_t(3),\mathbf{x}_t(4)]$ of the system~\eqref{eq:pendubot} and the remote estimation $\mathbf{\hat{x}}_t \triangleq [\mathbf{\hat{x}}_t(1),\mathbf{\hat{x}}_t(2),\mathbf{\hat{x}}_t(3),\mathbf{\hat{x}}_t(4)]$.}
	\label{fig:4process}
\end{figure}

Fig.~\ref{fig:compare} shows the simulated average estimation MSEs of the smart-sensor-based and a conventional-sensor-based remote estimator~\cite{schenato2008optimal} with different packet drop probabilities. Under the stability condition (illustrated as the gray area in Fig.~\ref{fig:4plots}(a)), we see that although the local estimator guarantees a better performance than the remote estimator, the performance gap is non-negligible only when the packet dropout probabilities are large at all channel states. It is interesting to note that the two cases actually have the same stability condition  when packet dropout are i.i.d., see~\cite{schenato2008optimal}. This motivates the hypothesis that the smart-sensor-based and the conventional-sensor-based remote estimation systems have the same stability condition under the Markov channel in terms of the LTI system transition matrix, the packet drop probabilities and the channel state transition matrix.

\begin{figure}[t]
	\centering\includegraphics[scale=0.6]{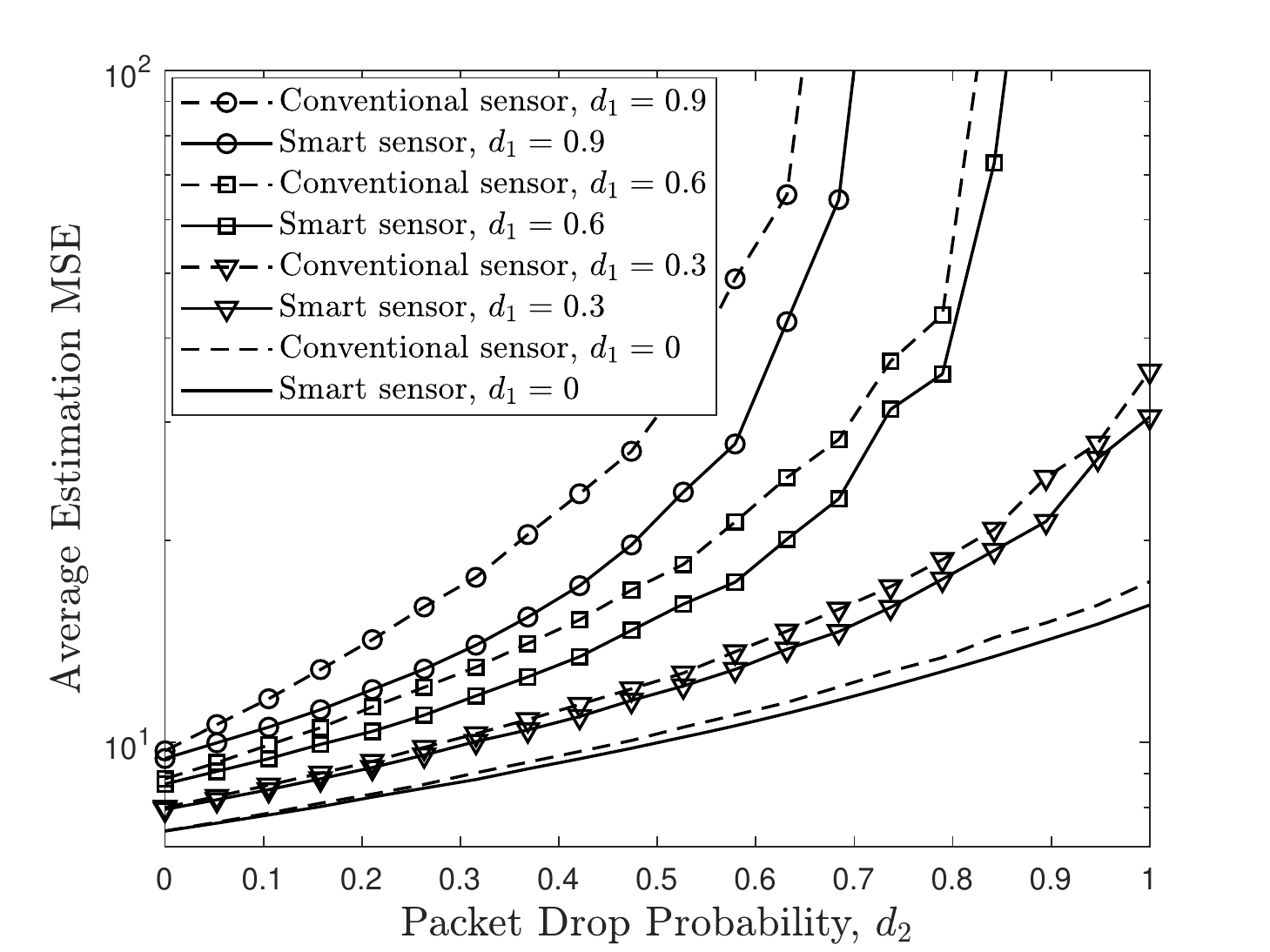}
	\caption{Average estimation MSE versus packet drop probabilities for the smart-sensor-based and conventional-sensor-based cases.}
	\label{fig:compare}
\end{figure}

In Figs.~\ref{fig:simu_contour_local} and~\ref{fig:simu_contour_remote}, we give contour plots of the average estimation MSE in the smart-sensor-based and the conventional-sensor-based cases, respectively. It can be observed that the average estimation MSEs grow up dramatically outside the theoretical stability region (i.e., Fig.~\ref{fig:4plots}(a)) in both cases. This verifies the correctness of the stability condition and also implies that the stability condition of the two cases can be the same.
\begin{figure}[t]
	\centering\includegraphics[scale=0.6]{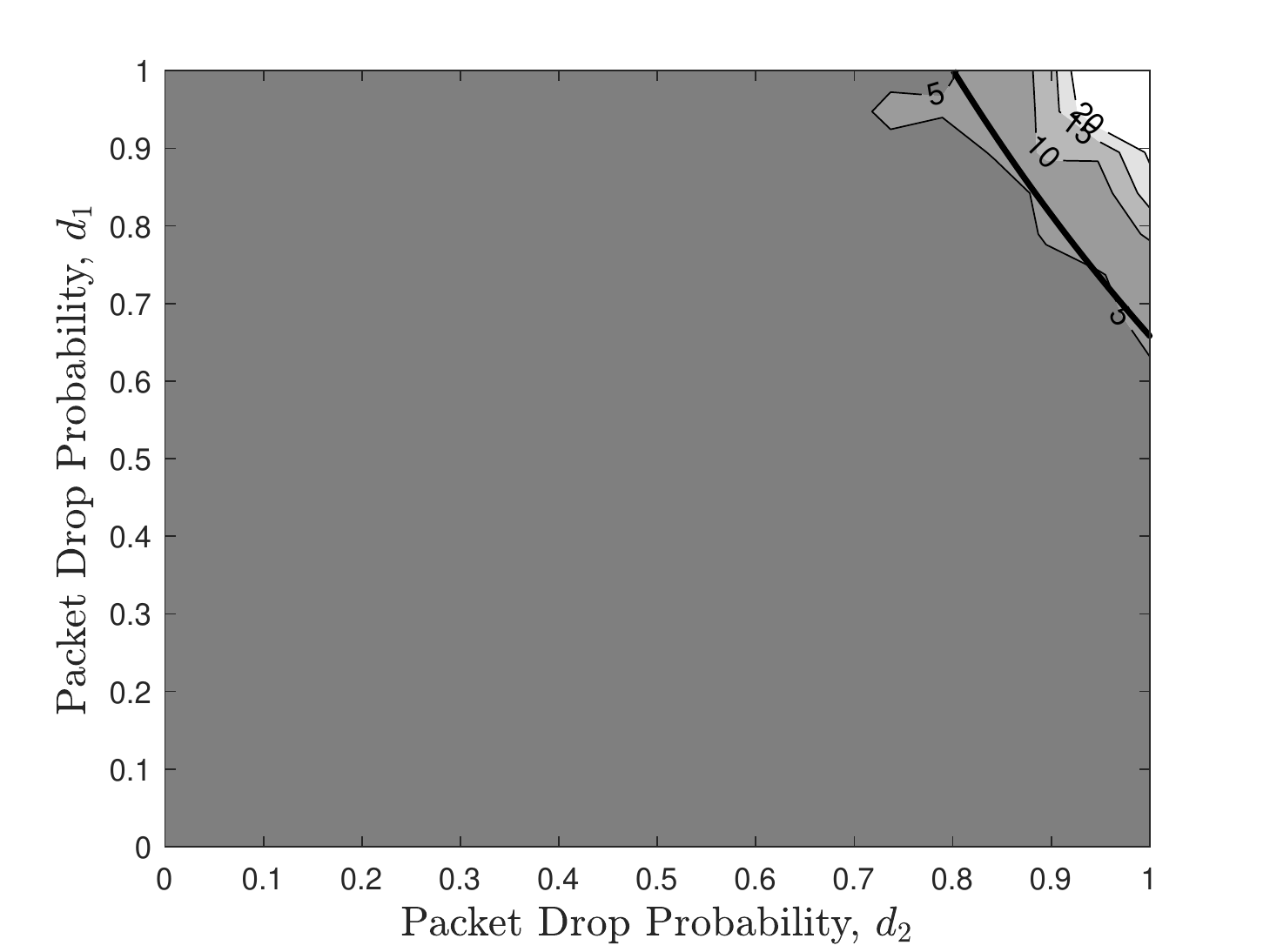}
	\caption{A contour plot (in a base-10 logarithmic scale) of the average estimation MSE of the smart-sensor-based case and the theoretical stability region (i.e., the thick line bounded region), i.e., Fig.~\ref{fig:4plots}(a).}
	\label{fig:simu_contour_local}
\end{figure}

\begin{figure}[t]
	\centering\includegraphics[scale=0.6]{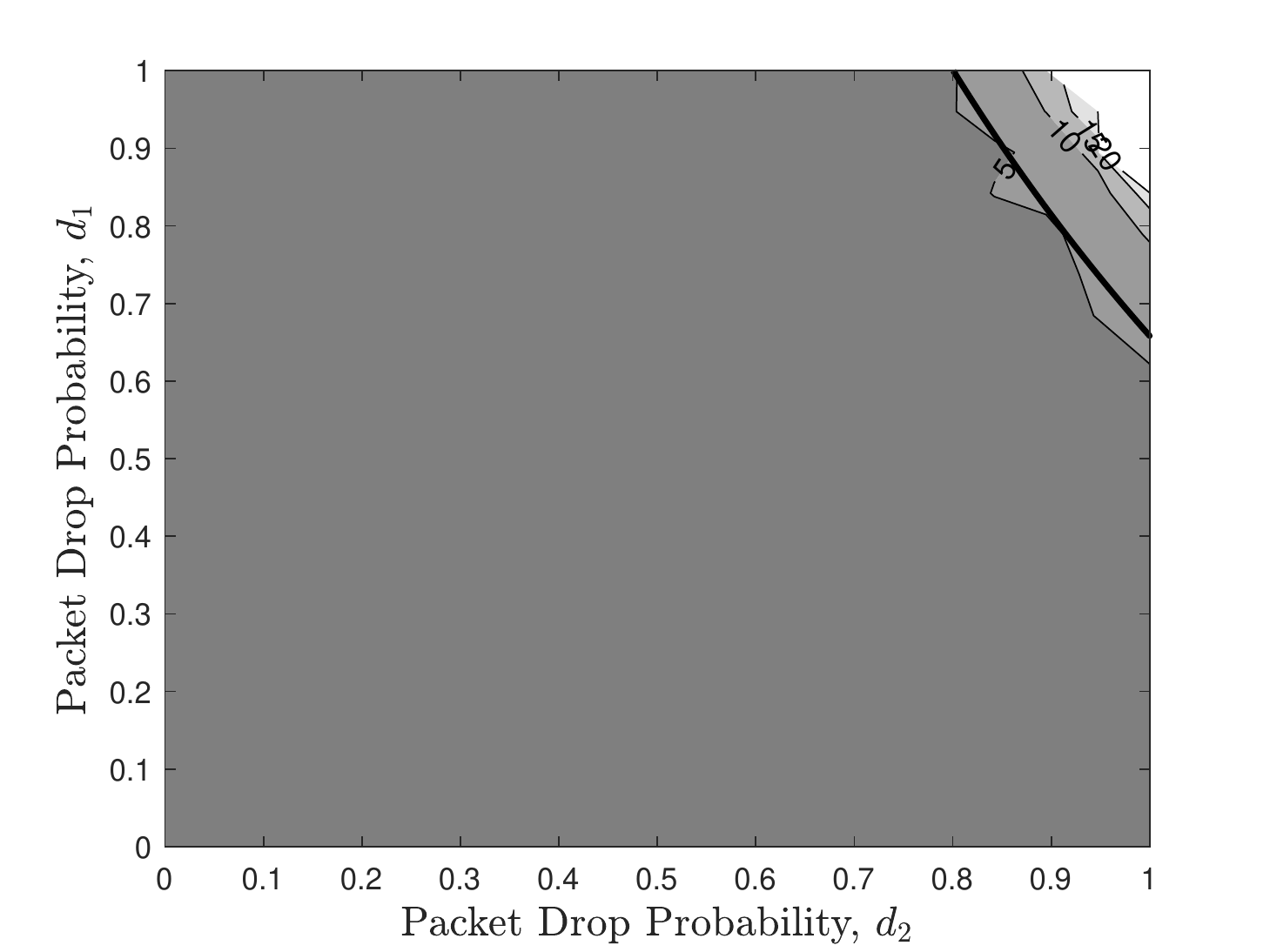}
	\caption{A contour plot (in a base-10 logarithmic scale) of the average estimation MSE of the conventional-sensor-based case and the theoretical stability region (i.e., the thick line bounded region), i.e., Fig.~\ref{fig:4plots}(a).}
	\label{fig:simu_contour_remote}
\end{figure}

In Fig.~\ref{fig:2versus3}, we consider a three-state Markov channel with the state transition matrix $\mathbf{M}_3 = \begin{bmatrix}
0.1 & 0.45 & 0.45\\
0.25 & 0.5 & 0.25\\
0.1& 0.1 & 0.8
\end{bmatrix}$.
Comparing with the two-state channel model considered in Fig.~\ref{fig:compare}, a third channel state with a small packet drop rate $d_3=0.01$ is added.
We see that the three-state channel case performs much better than the two-state channel one in terms of the average estimation MSE as the former channel is less error-prone.

\begin{figure}[t]
	\centering\includegraphics[scale=0.6]{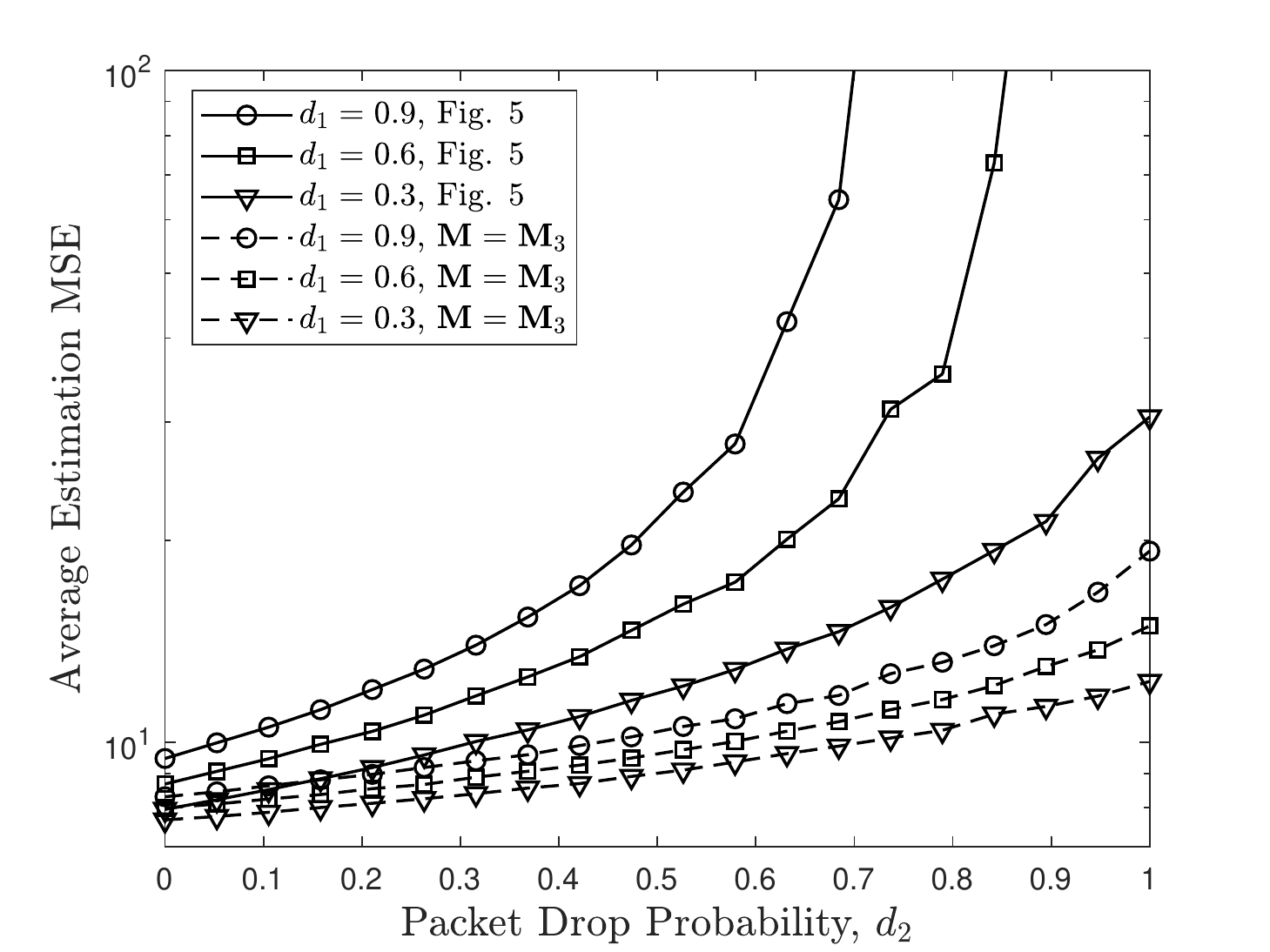}
	\caption{Average estimation MSE versus packet drop probabilities for the two channel state and the three channel state cases.}
	\label{fig:2versus3}
\end{figure}

\section{Conclusions} \label{sec:con}
In this paper, we have established the necessary and sufficient mean-square stability condition of a smart-sensor-based remote estimation system over a Markov fading channel, by developing the asymptotic theory of matrices which provides new (periodic) element-wise bounds of matrix powers. 
Our numerical results have verified the correctness of the stability condition and have shown that it is much more effective than   existing sufficient conditions in the literature. It has been observed that the stability region in terms of the packet drop probabilities in different channel states can either be convex or non-convex depending on the transition probability matrix. 
Our simulation results have suggested that the stability conditions may coincide for schemes with a  
smart sensor and with a conventional sensor. This inspires our future work on analyzing the stability condition of the case without a smart sensor.
Furthermore, the derived stability conditions can be used to design the optimal policy of transmission power control (e.g., via off-line design of $\mathbf{D}$) as well as multi-sensor scheduling policies.
In addition to stability, we will look into the performance of the smart sensor based remote estimation system and investigate the stationary distribution of estimation error covariances.

Furthermore, we will investigate the stability condition of the remote estimation system that adopts an aperiodic transmission policy. For example, an energy-constrained sensor normally uses an energy-saving protocol in practice that is to transmit only when the channel condition is good and/or the receiver has not received a packet for long time, i.e., with a large AoI state.
We will extend the stochastic-estimation-cycle-based analytical approach to this case. Specifically, when analyzing the lengths of the stochastic estimation cycles, we will take into account how the transmission decisions depend on the current AoI state and the current channel conditions.

\section*{Appendix A: Proof of Lemma~\ref{lem:G}}
	The time homogeneousty of $\{S\}_{\mathbb{N}_0}$ is clear due the time homogeneous Markov channel states.
	Let's define the following sets $\mathcal{I}_0\triangleq\{b_i:d_i=0,\forall i\in \mathcal{M}\}$, $\mathcal{I}_{0/1}\triangleq\{b_i:d_i\neq 0,1,\forall i\in \mathcal{M}\}$ and $\mathcal{I}_1\triangleq\{b_i:d_i=1,\forall i\in \mathcal{M}\}$ denoting the channel states in which a packet transmission must succeed, can succeed or fail, and must fail, respectively.
	Since $\mathbf{D} \neq \mathbf{0}$, the `can succeed' state set $\mathcal{I}_0 \cup \mathcal{I}_{0/1} \neq \emptyset$.
	Due to the ergodicity of the Markov channel, it can be proved that given any current channel state, the hitting time of any state of $\mathcal{I}_0 \cup \mathcal{I}_{0/1}$ is finite with a positive probability. Also,
	given a post-success state in $\mathcal{B}'$, it is reachable from a state of $\mathcal{I}_0 \cup \mathcal{I}_{0/1}$ in one step.
	Thus, given any state $S_k \in \mathcal{B}'$, the hitting time of any $S_{k+1} \in \mathcal{B}'$ is finite with a positive probability. 
	This completes the proof of the ergodicity of $\{S\}_{\mathbb{N}_0}$.
	Then, the stationary distribution $\boldsymbol{\beta}$ is the solution of 
	\begin{equation}
	\boldsymbol{\beta}^\top = \boldsymbol{\beta}^\top \mathbf{G}',
	\end{equation}
	where the state transition probability is 
	$\mathbf{G}'_{i,j} \triangleq \myprob{S_{k+1} = b_j \vert S_{k} = b_i},\forall i,j\in\mathcal{M}'$.
	Let $H_{k+1}\in \mathbb{N}$ denote the hitting time from $S_{k}$ to $S_{k+1}$, and $\mathbf{m}^\top_i\in\mathbb{R}^n$ and $\mathbf{n}_i\in\mathbb{R}^n$ denotes the $i$th row and $i$th column of the matrix $\mathbf{M}$. 
	We further have
	\begin{align}
	&\mathbf{G}'_{i,j} 
	=\sum_{l=1}^{\infty} \myprob{S_{k+1} = b_j, H_{k+1}=l \vert S_{k} = b_i}\\
	&= \!(1\!-d_i) p_{i,j} \!+ d_i \mathbf{m}^\top_i \mathbf{(I-D)} \mathbf{n}_i 
	\!+ d_i \mathbf{m}^\top_i \mathbf{D M} \mathbf{(I-D)} \mathbf{n}_i 
	\!+ \cdots
	\end{align}
	which completes the proof of \eqref{eq:G}.
	From the definition of $\mathcal{B}'$ in \eqref{eq:B'}, it is clear that the last $(M-M')$ columns of $\mathbf{(I-D)M}$ are all zeros, completing the proof of Lemma~\ref{lem:G}.

\section*{Appendix B: Proofs of Lemmas~\ref{lem:matrix_up} and~\ref{lem:matrix_low_1}}

\subsection{Preliminaries}
Assume a $z$-by-$z$ matrix $\mathbf{Z}$ has $A$ different eigenvalues $\{\lambda_1,\lambda_2,\cdots,\lambda_A\}$.	Represent $\mathbf{Z}$ in its Jordan normal form of $\mathbf{Z} =\mathbf{U} \mathbf{J} \mathbf{U}^{-1}$, where $\mathbf{U}$ is a invertible matrix and 
\begin{equation}
\mathbf{J} = \begin{bmatrix}
\mathbf{J_1} & &\\
&\ddots &\\
& & \mathbf{J}_A\\
\end{bmatrix},
\end{equation}
\begin{equation}
\mathbf{J}_m = \begin{bmatrix}
\lambda_m & 1 & &\\
&\lambda_m & \ddots & \\
& & \ddots&  1\\
& & & \lambda_m\\
\end{bmatrix},\forall m=1,\cdots,A.
\end{equation}
Let $u_m$ denote the size of Jordan block $    \mathbf{J}_m$.
Then, $\mathbf{U}$ and $\mathbf{U}^{-1}$ can be represented as 
\begin{equation}
\begin{aligned}
\mathbf{U} &= \left[\mathbf{F}_1 \vert \mathbf{F}_2 \vert \cdots \vert \mathbf{F}_A \right]\\
\mathbf{U}^{-1} &= \left[\mathbf{G}_1 \vert \mathbf{G}_2 \vert \cdots \vert \mathbf{G}_A \right]^\top
\end{aligned}
\end{equation}
where $\mathbf{F}_m$ and $\mathbf{G}_m$ are $z$-by $u_m$ matrices.
Since $\mathbf{U}$ is of full rank, $\mathbf{F}_m$ and $\mathbf{G}_m$ have full column rank of $u_m, \forall m=1,\cdots,A$.

Then, we have
\begin{equation}\label{A_power}
\mathbf{Z}^i = \mathbf{U} \mathbf{J}^i \mathbf{U}^{-1} = \sum_{m=1}^{A} \mathbf{F}_m \mathbf{J}^i_m {\mathbf{G}^\top_m},
\end{equation}
where
\begin{equation}\label{jordan}
\mathbf{J}^i_m  = \begin{bmatrix}
& \lambda^i_m & \binom{i}{1}\lambda^{i-1}_m & \cdots & \cdots & \binom{i}{u_m-1}\lambda^{i-u_m+1}_m \\
&  & \ddots & \ddots & \vdots & \vdots\\
&  & & \ddots & \ddots & \vdots\\
&  & &  & \lambda^i_m & \binom{i}{1}\lambda^{i-1}_m\\
&  &  &  &  & \lambda^i_m
\end{bmatrix}.
\end{equation}
Note that ${\mathbf{J}^i_m}$  has a full rank of $u_m$ for all $m\in\mathbb{N}$ if $\lambda_m \neq 0$.

From \eqref{A_power} and \eqref{jordan}, the element at the $j$th row and $k$th column of $\mathbf{F}_i \mathbf{J}^m_i {\mathbf{G}^\top_i}$ denoted by $[\mathbf{F}_m  \mathbf{J}^i_m \mathbf{G}^\top_m]_{j,k}$, can be rewritten as a polynomial in terms of $m$ as
\begin{equation}\label{dominant}
[\mathbf{F}_m  \mathbf{J}^i_m \mathbf{G}^\top_m]_{j,k}= \lambda^i_m [i^{u_m-1}, i^{u_m-2},\cdots,i,1 ]\mathbf{\Lambda}_{m,(j,k)},
\end{equation}
where $\mathbf{\Lambda}_{m,(j,k)}$ is a column vector determined by $\mathbf{F}_m$ and $\mathbf{G}_m$ and is independent with $i$.

\subsection{Proof of Lemma~\ref{lem:matrix_up}}
From \eqref{A_power} and \eqref{dominant}, we have
\begin{equation}
\vert [\mathbf{Z}^i]_{j,k}\vert =\left\lvert \sum_{m=1}^{A} [\mathbf{F}_m  \mathbf{J}^i_m \mathbf{G}^\top_m]_{j,k} \right\rvert \leq \kappa i^z \sum_{m=1}^{A} \vert \lambda_m \vert^i \leq A \kappa  i^z \rho^i(\mathbf{Z}),
\end{equation}
where $\kappa$ is a positive constant.

The result follows upon noting that $\lim\limits_{i\rightarrow \infty} i^z \rho^i(\mathbf{Z})/(\rho(\mathbf{Z})+\epsilon)^i = 0, \forall \epsilon>0$.

\subsection{Proof of Lemma~\ref{lem:matrix_low_1}}
Before proceeding to prove the element-wise lower bound of matrix powers, we need the following technical lemma.
\begin{lemma}\label{lem:tech}
	\normalfont
	Consider a $z$-by-$z$ matrix $\mathbf{Z}$ with $A$ different eigenvalues $\{\lambda_1,\lambda_2,\cdots,\lambda_A\}$, where $A\leq z$, and a $z$-by-$z$ symmetric and positive semidefinite matrix $\mathbf{Q}$.	
	If $(\mathbf{Z},\sqrt{\mathbf{Q}})$ is controllable and $\lambda_m \neq 0$ for some $m \in \{1,\cdots,A\}$, then $\mathbf{F}_m \mathbf{J}^i_m \mathbf{G}^{\top}_m \sqrt{\mathbf{Q}}\neq \mathbf{0}, \forall i \in \mathbb{N}_0$.
\end{lemma}

\begin{proof}
	Assume that there exists $i \in \mathbb{N}_0$ such that $\mathbf{F}_m \mathbf{J}^i_m \mathbf{G}^{\top}_m \sqrt{\mathbf{Q}} = \mathbf{0}$.
	Then, we have
	\begin{align}
	&0 = \myrank{\mathbf{F}_m \mathbf{J}^i_m \mathbf{G}^{\top}_m \sqrt{\mathbf{Q}}} \\ \label{Sylvester}
	&\geq \myrank{\mathbf{F}_m} + \myrank{\mathbf{J}^i_m \mathbf{G}^{\top}_m \sqrt{\mathbf{Q}}} -u_m\\
	& =\myrank{\mathbf{J}^i_m \mathbf{G}^{\top}_m \sqrt{\mathbf{Q}}}\\
	& =\myrank{\mathbf{G}^{\top}_m \sqrt{\mathbf{Q}}},
	\end{align}
	where \eqref{Sylvester} is due to Sylvester's rank inequality~\cite{strang1993introduction}.
	Thus, $\mathbf{G}^{\top}_m \sqrt{\mathbf{Q}} = \mathbf{0}$.
	
	From the definitions of $\mathbf{F}_m$ and $\mathbf{G}_m$, it is clear that \begin{equation} \label{identity}
	\sum_{k=1}^{A} \mathbf{F}_k \mathbf{G}^\top_k = \mathbf{I}.
	\end{equation}
	By multiplying $\sqrt{\mathbf{Q}}$ on the both sides of \eqref{identity}, we have
	\begin{equation}\label{Q}
	\sqrt{\mathbf{Q}} = \sum_{k\in \{1,\cdots,A\}\backslash m} \mathbf{F}_k \mathbf{G}^\top_k \sqrt{\mathbf{Q}}.
	\end{equation}
	
	Applying $\mathbf{G}^{\top}_m \sqrt{\mathbf{Q}} = \mathbf{0}$ on \eqref{A_power}, it can be obtained that
	\begin{equation}\label{AQ}
	\mathbf{Z}^i \sqrt{\mathbf{Q}} = \sum_{k\in \{1,\cdots,A\}\backslash m} \mathbf{F}_k \mathbf{J}^i_k {\mathbf{G}^\top_k} \sqrt{\mathbf{Q}}, \forall i \in \mathbb{N}_0.
	\end{equation}
	Jointly using \eqref{Q} and \eqref{AQ}, it is easy to see that each column of the matrix concatenation $\left[\sqrt{\mathbf{Q}},\mathbf{Z}\sqrt{\mathbf{Q}},\cdots,\mathbf{Z}^{z-1}\sqrt{\mathbf{Q}}\right]$ is in the span of the columns of $\{\mathbf{F}_k\vert k\in \{1,\cdots,A\}\backslash m \}$. Therefore,
	\begin{equation}
	\begin{aligned}
	&\myrank{\left[\sqrt{\mathbf{Q}},\mathbf{Z}\sqrt{\mathbf{Q}},\cdots,\mathbf{Z}^{z-1}\sqrt{\mathbf{Q}}\right]}\\
	&\leq \sum_{k\in \{1,\cdots,A\}\backslash m} \myrank{\mathbf{F}_k}\\
	&=z-\myrank{\mathbf{F}_m}\\
	&<z,
	\end{aligned}
	\end{equation}
	which, however, contradicts with the assumption that $\left[\sqrt{\mathbf{Q}},\mathbf{Z}\sqrt{\mathbf{Q}},\cdots,\mathbf{Z}^{z-1}\sqrt{\mathbf{Q}}\right]$ is of full rank.
	This completes the proof.
\end{proof}

\begin{proof}[Proof of Lemma~\ref{lem:matrix_low_1}(i)]
	If $\lambda_m\neq 0$, $\mathbf{J}_m$ is a full-rank square matrix and hence $\mathbf{J}^i_m \mathbf{G}^\top_m$ has a full row rank of $u_m$.
	Since $\mathbf{F}_m$ has a full column rank of $u_m$, by using Sylvester's rank inequality, we have
	\begin{equation}
	\begin{aligned}
	\text{Rank}\left(\mathbf{F}_m \mathbf{J}^i_m {\mathbf{G}^\top_m}\right)
	&\geq \text{Rank}\left(\mathbf{F}_m\right)
	+ \text{Rank}\left(\mathbf{J}^i_m {\mathbf{G}^\top_m}\right)-u_m \\
	&= u_m>0.
	\end{aligned}
	\end{equation}
	Therefore, $\mathbf{F}_m \mathbf{J}^i_m {\mathbf{G}^\top_m} \neq \mathbf{0}, \forall i \in \mathbb{N}$.
	From \eqref{dominant}, we can find a pair of $j,k\in\mathcal{Z}$ such that $\mathbf{\Lambda}_{m,(j,k)}\neq \mathbf{0}$.
	Thus, without loss of generality, we assume that the dominant term of the polynomial $[\mathbf{F}_m  \mathbf{J}^i_m \mathbf{G}^\top_m]_{j,k}=\lambda^i_m [i^{u_m-1}, i^{u_m-2},\cdots,i,1 ]\mathbf{\Lambda}_{m,(j,k)}$ is $\Lambda_{m,(j,k)} \lambda^i_m i^{u_{m,(j,k)}}$ when $i\rightarrow \infty$, where $\Lambda_{m,(j,k)}\neq 0$ and $u_{m,(j,k)}\in\{0,\cdots,u_m-1\}$.	
	
	If $\vert \lambda_m \vert = \rho(\mathbf{Z})$ and $\lambda_m $ is the unique eigenvalue that has the maximum magnitude, it is clear that the dominant term of 
	$[\mathbf{Z}^i]_{j,k}=\sum_{m=1}^{\tilde{z}}[\mathbf{F}_m  \mathbf{J}^i_m \mathbf{G}^\top_m]_{j,k}$ is $\Lambda_{m,(j,k)} \lambda^i_m i^{u_{m,(j,k)}}$. Thus, one can find $\eta>0$ such that $\vert [\mathbf{Z}^i]_{j,k} \vert^2$ is asymptotically lower bounded by $\eta \rho^{2i}(\mathbf{Z})$.
	
	If there are multiple eigenvalues having the same maximum magnitude, i.e.,
	$\mathcal{Z}' \triangleq \{i: \vert \lambda_i \vert = \rho(\mathbf{Z}),  \forall i\in \mathcal{Z}\}$ and $\vert \mathcal{Z}'\vert >1$, where $\mathcal{Z}\triangleq\{1,2,\cdots,A\}$, 
	we consider the following two complementary cases:
	
	Case 1): There exists $m\in\mathcal{Z'}$ such that $\Lambda_{m,(j,k)}\neq 0 $ and $u_{m,(j,k)}>u_{m',(j,k)},\forall m' \in \mathcal{Z'}\backslash\{m\}$. In this case, the dominant term of $[\mathbf{Z}^i]_{j,k}=\sum_{m=1}^{\tilde{z}}[\mathbf{F}_m  \mathbf{J}^i_m \mathbf{G}^\top_m]_{j,k}$ is still $\Lambda_{m,(j,k)} \lambda^i_m i^{u_{m,(j,k)}}$.
	Thus, one can find $\eta>0$ such that $\vert [\mathbf{Z}^i]_{j,k} \vert^2$ is asymptotically lower bounded by $\eta \rho^{2i}(\mathbf{Z})$.
	
	Case 2): There exists a set $\mathcal{Z}'' \subseteq \mathcal{Z}'$ with cardinality $z''\geq 2$ such that  $u_{m,(j,k)}=u_{m',(j,k)},\forall m,m'\in\mathcal{Z}''$ and $u_{m,(j,k)}>u_{m',(j,k)},\forall m\in\mathcal{Z}'',m'\in\mathcal{Z}'\backslash \mathcal{Z}''$.
	In this case,  $\vert [\mathbf{Z}^i]_{j,k} \vert^2$ may not be asymptotically lower bounded by $\eta \rho^{2i}(\mathbf{Z})$ due to the multiple eigenvalues with identical magnitude but different phases.
	
	In the following, we will show that $\sum_{m\in \mathcal{Z''}} \Lambda_{m,(j,k)} \lambda^i_m i^{u_{m,(j,k)}}$ is asymptotically and periodically bounded by $\eta \rho^{2i}(\mathbf{Z})$.
	Let $m_l$ denote the index of the $l$th eigenvalue in $\mathcal{Z''}$, where $l\in\{1,\cdots,z''\}$. Thus,  $\lambda_{m_l} \triangleq \rho(\mathbf{Z}) e^{\mathrm{j}\phi_{m_l}}$, where $\phi_{m_l}\in[0,2\pi)$.
	We have the following matrix
	\begin{equation} 
	\begin{aligned}
	\mathbf{\Pi} 
	&\triangleq \begin{bmatrix}
	\lambda^i_{m_1} & \lambda^i_{m_2}&\cdots&\lambda^i_{m_{z''}}\\
	\lambda^{i+1}_{m_1} & \lambda^{i+1}_{m_2}&\cdots&\lambda^{i+1}_{m_{z''}}\\
	\vdots & \vdots&\cdots&\vdots\\
	\lambda^{i+z''-1}_{m_1} & \lambda^{i+z''-1}_{m_2}&\cdots&\lambda^{i+z''-1}_{m_{z''}}
	\end{bmatrix} \\
	&= \text{diag}\left\lbrace\rho^i(\mathbf{Z}) e^{\mathrm{j} i \phi_{m_1}},\rho^{i+1}(\mathbf{Z}) e^{\mathrm{j} {(i+1)} \phi_{m_1}},\cdots,\right.\\
	&\left. \qquad\qquad \rho^{i+z''-1}(\mathbf{Z}) e^{\mathrm{j} (i+z''-1) \phi_{m_1}} \right\rbrace \mathbf{\Pi}',
	\end{aligned}
	\end{equation}
	where 
	\begin{equation}\notag
	\begin{aligned}
	\mathbf{\Pi}' 
	&\triangleq \begin{bmatrix}
	1 & e^{\mathrm{j} i (\phi_{m_2}-\phi_{m_1})}&\cdots&e^{\mathrm{j} i (\phi_{m_{z''}}-\phi_{m_1})}\\
	1 & e^{\mathrm{j} (i+1) (\phi_{m_2}-\phi_{m_1})}&\cdots&e^{\mathrm{j} (i+1) (\phi_{m_{z''}}-\phi_{m_1})}\\
	\vdots & \vdots&\cdots&\vdots\\
	1 & e^{\mathrm{j} (i+z''-1) (\phi_{m_2}-\phi_{m_1})}&\cdots&e^{\mathrm{j} (i+z''-1) (\phi_{m_{z''}}-\phi_{m_1})}\\
	\end{bmatrix}\\
	&=\mathbf{\Pi}'' \mathbf{\Phi} \\
	&\triangleq\begin{bmatrix}
	1 & e^{\mathrm{j}  (\phi_{m_2}-\phi_{m_1})}&\cdots&e^{\mathrm{j}  (\phi_{m_{z''}}-\phi_{m_1})}\\
	1 & e^{\mathrm{j} 2 (\phi_{m_2}-\phi_{m_1})}&\cdots&e^{\mathrm{j} 2 (\phi_{m_{z''}}-\phi_{m_1})}\\
	\vdots & \vdots&\cdots&\vdots\\
	1 & e^{\mathrm{j} z'' (\phi_{m_2}-\phi_{m_1})}&\cdots&e^{\mathrm{j} z'' (\phi_{m_{z''}}-\phi_{m_1})}
	\end{bmatrix} \times\\	
	\end{aligned}
	\end{equation}
	\begin{equation}
	\begin{aligned}	
	& \quad \begin{bmatrix}
	1 & 0 &\cdots&0\\
	0 & e^{\mathrm{j} (i-1) (\phi_{m_2}-\phi_{m_1})}&\cdots& 0\\
	\vdots & \vdots&\cdots&\vdots\\
	0 & 0 &\cdots&e^{\mathrm{j} (i-1) (\phi_{m_{z''}}-\phi_{m_1})}
	\end{bmatrix},
	\end{aligned}
	\end{equation}
	and $\mathbf{\Pi}''$ is a Vandermonde matrix, which is invertible  due to the fact that $\lambda_{m_l}\neq\lambda_{m_l'}$ for $l\neq l'$, see~\cite{Vandermonde}.
	Let $\mathbf{b}\triangleq [\Lambda_{m_1,(j,k)},\Lambda_{m_2,(j,k)},\cdots,\Lambda_{m_{z''},(j,k)}]^\top \neq \mathbf{0}$. Since $\mathbf{\Pi}''$ is invertible, using the inequality of matrix-vector product~\cite{norm1,norm2}, we have 
	\begin{equation}
	\vert \mathbf{\Pi}'\mathbf{b} \vert = \vert \mathbf{\Pi}''\mathbf{\Phi} \mathbf{b} \vert \geq \vert (\mathbf{\Pi}'')^{-1} \vert^{-1} \vert \mathbf{\Phi} \mathbf{b} \vert =\vert (\mathbf{\Pi}'')^{-1} \vert^{-1} \vert \mathbf{b} \vert >0,
	\end{equation}
	where $\vert (\mathbf{\Pi}'')^{-1} \vert^{-1} \neq 0$ is the minimum magnitude of the eigenvalues of $\mathbf{\Pi}''$.
	Since the largest magnitude of the elements of $\mathbf{\Pi}'\mathbf{b}$ is no smaller than $\vert \mathbf{\Pi}'\mathbf{b} \vert/\sqrt{z''}$, we have
	\begin{equation}
	\begin{aligned}
	& \max_{l' = 0,\cdots,z''-1 } \left\lvert \sum_{l=1}^{z''}\Lambda_{m_l,(j,k)} e^{\mathrm{j}(i+l')(\phi_{m_l-m_1})} \right\rvert \\
	&\geq \frac{1}{\sqrt{z''}} \vert (\mathbf{\Pi}'')^{-1} \vert^{-1} \vert \mathbf{b} \vert>0,
	\end{aligned}
	\end{equation}
	and hence
	\begin{equation}
	\begin{aligned}
	&\max_{l' = 0,\cdots,z''-1 } \left\lvert \sum_{l=1}^{z''}\Lambda_{m_{l},(j,k)} \lambda^{(i+l')}_{m_{l}} (i+l')^{u_{m_{l},(j,k)}} \right\rvert \\
	&\geq \frac{1}{\sqrt{z''}} \vert (\mathbf{\Pi}'')^{-1} \vert^{-1} \vert \mathbf{b} \vert \rho^i(\mathbf{Z}) i^{u_{m_1,(j,k)}}.
	\end{aligned}
	\end{equation}
	Therefore, $\vert [\mathbf{Z}^i]_{j,k} \vert^2$ is asymptotically and periodically lower bounded by $\eta \rho^{2i}(\mathbf{Z})$ with period $z''$, where $\eta$ is a positive constant.	
\end{proof}

\begin{proof}[Proof of Lemma~\ref{lem:matrix_low_1}(ii)]
	From \eqref{dominant}, if $\lambda_m \neq 0$, it is easy to see that
	\begin{equation}
	[\mathbf{F}_m  \mathbf{J}^i_m \mathbf{G}^\top_m \sqrt{\mathbf{Q}}]_{j,k}=\lambda^i_m [i^{u_m-1}, i^{u_m-2},\cdots,i,1 ]\mathbf{\Lambda}'_{m,(j,k)},
	\end{equation}
	where $\mathbf{\Lambda}'_{m,(j,k)}$ is a column vector determined by $\mathbf{F}_m$, $\mathbf{G}_m$ and $\sqrt{\mathbf{Q}}$ and is independent with~$i$.
	Thus, there exist $j,k\in\mathcal{Z}$ such that $\mathbf{\Lambda}'_{m,(j,k)}\neq \mathbf{0}$; otherwise, it violates Lemma~\ref{lem:tech}.
	Then, by following the same steps of the proof of Lemma~\ref{lem:matrix_low_1}(i), we can prove that there exists $\eta>0$ such that $\vert [\mathbf{Z}^i \sqrt{\mathbf{Q}}]_{j,k} \vert^2$ is asymptotically and periodically lower bounded by $\eta \rho^{2i}(\mathbf{Z})$.
\end{proof}

\section*{Appendix C: Proofs of Lemmas~\ref{lem:matrix_low_2'} and~\ref{lem:matrix_low_2}}
\subsection{Proof of Lemma~\ref{lem:matrix_low_2'}}
	For part (i), 
since $\mathbf{M}$ is an irreducible non-negative matrix and $d_j>0, \forall j\in \mathcal{M}$, the $M$-by-$M$ matrix $\mathbf{DM}$ is also irreducible and non-negative, i.e., one can associate with the matrix a certain directed graph $\mathcal{G}$, which has exactly $M$ vertexes, and there is an edge from vertex $j$ to vertex $k$ precisely when $[\mathbf{DM}]_{j,k} > 0$, and $\mathcal{G}$ is strongly connected.
By using Lemma~\ref{lem:matrix_low_1}, there exists $j,k\in\mathcal{M}$ such that $[(\mathbf{DM})^i]_{j,k}$ is asymptotically and periodically lower bounded by $\eta \rho^i(\mathbf{DM})$, where the period is no larger than the number of eigenvalues of $\mathbf{DM}$ with the same maximum magnitude.
Then using the non-negative and irreducible property of $\mathbf{DM}$, given $k'\in\mathcal{M}$, we can find $l \in \mathbb{N}^+$ such that $k'$ is reachable from $k$ in $l$ steps, i.e., $[(\mathbf{DM})^l]_{k,k'} = \eta' >0$. Thus, if $[(\mathbf{DM})^i]_{j,k} \geq \eta \rho^i(\mathbf{DM})$,  we have $$[(\mathbf{DM})^{i+l}]_{j,k'} \geq [(\mathbf{DM})^{i}]_{j,k} (\mathbf{DM})^{l}]_{k,k'} \geq \frac{\eta \eta'}{\rho^l(\mathbf{DM})} \rho^{i+l}(\mathbf{DM}).$$
Since  $[(\mathbf{DM})^i]_{j,k}$ is asymptotically and periodically lower bounded by $\eta \rho^i(\mathbf{DM})$, $[(\mathbf{DM})^{i}]_{j,k'}$ is asymptotically and periodically lower bounded by $\frac{\eta \eta'}{\rho^l(\mathbf{DM})} \rho^{i}(\mathbf{DM})$.
This completes the proof of part (i).

For part (ii), we construct a diagonal matrix $\mathbf{D}' \triangleq \text{diag}\{d'_1,\cdots,d'_M\}$, where $d'_j=d_j$ if $d_j>0$ otherwise $d'_j=1$. Thus, $\mathbf{D'M}$ is irreducible and non-negative.
Since $\mathbf{DM}$ has all zero rows, the direct graph of $\mathbf{DM}$, $\mathcal{G}$, can be generated by the strongly connected graph $\mathcal{G}'$ induced by $\mathbf{D'M}$ and then remote the edges from vertexes $k\in\mathcal{J}_0$. In other words, $\mathcal{G}$ is part of $\mathcal{G}'$. Then, it is easy to see that for any $k\in\mathcal{\bar{J}}_0$, we can find $k'\in \mathcal{J}_0$ such that there exists a path from $k$ to $k'$ in $\mathcal{G}$, otherwise, it violates the irreducible property of $\mathcal{G}'$.
Therefore, for any $k\in \mathcal{\bar{J}}_0$, there exist $k'\in \mathcal{J}_0$ and  $l\in \mathbb{N}$ such that $[(\mathbf{DM})^l]_{k,k'}>0$.
By using this property, part (ii) of Lemma~\ref{lem:matrix_low_2'} can be proved by following the similar steps of part (i).

For part (iii), since $[(\mathbf{DM})^i]_{j,k} \leq [\mathbf{M}^i]_{j,k}$, $\forall j,k \in \mathcal{M}$ and each element of $\mathcal{M}$ is strictly less than $1$, we have $[(\mathbf{DM})^i]_{j,k} <1$, $\forall j,k \in \mathcal{M}$. By using part (ii), i.e., there exist $j,k \in \mathcal{M}$ and $\eta>0$ such that $[(\mathbf{DM})^i]_{j,k}$ is asymptotically and periodically lower bounded by $\eta \rho^i(\mathbf{DM})$, we have $\rho(\mathbf{DM})<1$.

\subsection{Proof of Lemma~\ref{lem:matrix_low_2}}
For the case that $\mathcal{J}_0=\emptyset$, since $\mathbf{D}\neq \mathbf{I}$ and $\mathbf{M}$ is a stochastic matrix, there exists $k',k \in \mathcal{M}$ such that $[\mathbf{(I-D)M}]_{k',k}=\eta'>0$.
Using Lemma~\ref{lem:matrix_low_2'}(i), for any $j,k'\in\mathcal{M}$, $[(\mathbf{DM})^i]_{j,k'}$ is asymptotic and periodically lower bounded by $\eta \rho^i(\mathbf{DM})$. Thus, there exists $j,k',k\in\mathcal{M}$ such that  $$[(\mathbf{DM})^i\mathbf{(I-D)M}]_{j,k}\geq [(\mathbf{DM})^i]_{j,k'} [\mathbf{(I-D)M}]_{k',k} \geq \eta'\eta \rho^i(\mathbf{DM}).$$

For the case that $\mathcal{J}_0\neq \emptyset$,  given any $k'\in \mathcal{J}_0$, we can find $k\in \mathcal{M}$ such that $[\mathbf{(I-D)M}]_{k',k}=\eta'>0$.
Using Lemma~\ref{lem:matrix_low_2'}(ii), there exists $j\in\mathcal{\bar{J}}_0,k'\in\mathcal{J}_0$ such that $[(\mathbf{DM})^i]_{j,k'}$ is asymptotic and periodically lower bounded by $\eta \rho^i(\mathbf{DM})$. Thus, there exists $j\in\mathcal{\bar{J}}_0$, $k' \in \mathcal{J}_0$ and $k \in \mathcal{M}$ such that, also in this case,  $[(\mathbf{DM})^i\mathbf{(I-D)M}]_{j,k}\geq [(\mathbf{DM})^i]_{j,k'} [\mathbf{(I-D)M}]_{k',k}\geq \eta'\eta \rho^i(\mathbf{DM})$.

\section*{Appendix D: Proof of Theorem~\ref{theorem:main}}

Using the property that $c(i)$ is a monotonically increasing function~\cite{shi2012scheduling}, from \eqref{g_fun}, we have
\begin{equation}\label{g_inequal}
c(i) \leq g(i) \leq i c(i),\forall i\in\mathbb{N}.
\end{equation}
Then, we consider two scenarios: (i) all channel states are post-success states, i.e., $\mathcal{B}'=\mathcal{B}$ and $M'=M$; and (ii) some channel states are not post-success states, i.e., $\mathcal{B}\backslash \mathcal{B}'\neq \emptyset$ and $M'<M$.

(i) $M=M'$. Using Proposition~\ref{lem:cost_up}, Lemma~\ref{lem:matrix_up} and the inequality~\eqref{g_inequal}, for any $\epsilon>0$, we can find $\kappa>0$ such that $\myexpect{C}$ in \eqref{C} is upper bounded as
\begin{equation} \label{C_up}
\myexpect{C}< \kappa M^2  \bar{\beta} \sum_{i=1}^{\infty} i (\rho(\mathbf{A})+\epsilon)^{2i} (\rho(\mathbf{DM})+\epsilon)^{i}+\gamma,
\end{equation}
where $\bar{\beta} \triangleq \max\{\beta_1,\cdots,\beta_M\}$, and $\gamma$ is a constant.
Thus, $\myexpect{C}$ is bounded if $\rho^2(\mathbf{A})\rho(\mathbf{DM})<1$.
By using Proposition~\ref{lem:cost_low} and Lemma~\ref{lem:matrix_low_2} and after some algebraic manipulations, there exists $\eta>0$ such that $\myexpect{C}$ in \eqref{C}
\begin{equation} \label{C_low}
\myexpect{C} > \eta \underline{\beta} \sum_{i=1}^{\infty} \rho^{2i}(\mathbf{A}) \rho^{i}(\mathbf{DM}) + \gamma',
\end{equation}
where $\underline{\beta} \triangleq \min\{\beta_1,\cdots,\beta_M\}>0$, and $\gamma'$ is a constant. 
Thus, $\rho^2(\mathbf{A})\rho(\mathbf{DM})<1$ if $\myexpect{C}$ is bounded. Therefore, $\myexpect{C}$ is bounded if and only if $\rho^2(\mathbf{A})\rho(\mathbf{DM})<1$.
Similarly, since $\rho(\mathbf{DM})<1$ given in Lemma~\ref{lem:matrix_low_2'}(iii), it can be proved that $\myexpect{T}$ is always bounded.

Therefore, $\myexpect{C}/\myexpect{T}$ is bounded if and only if $\rho^2(\mathbf{A})\rho(\mathbf{DM})<1$ holds. This completes the proof of Theorem~\ref{theorem:main} in scenario (i).

(ii) $M>M'$. 
It is clear that the upper bound in scenario (ii) is the same as in \eqref{C_up}. 
Different from scenario (i), the lower bound cannot be obtained as in \eqref{C_low} directly since $\underline{\beta} = \min\{\beta_1,\cdots,\beta_{M'},\underbrace{0,\cdots,0}_{M-M'}\}=0$ making the lower bound useless in scenario (ii).

Taking Proposition~\ref{lem:cost_low} and \eqref{g_inequal} into \eqref{C}, there exists $\eta>0$ such that
\begin{equation}\label{eq:C_case2}
\begin{aligned}
\myexpect{C}&>\eta\sum_{j=1}^{M}\sum_{k=1}^{M} \sum_{i=1}^{\infty} \rho^{2i}(\mathbf{A}) \left[\text{diag}\{\beta_1,\cdots,\beta_{M'},\underbrace{0,\cdots,0}_{M-M'}\} \right. \\
&\left. \qquad \times \mathbf{(DM)}^{i-1} \mathbf{(I-D)M}\right]_{j,k} + \gamma'',
\end{aligned}
\end{equation}
where $\gamma''$ is a constant.
Due to the ergodicity of the Markov channel, for any non-post-success state $b_{j'}$ in $\mathcal{B}\backslash\mathcal{B}'=\{b_{M'+1},\cdots,b_{M}\}$, there exists a post-success state $b_{i'}$ in $\mathcal{B}'$ such that $b_{i'}$ can transit to $b_{j'}$ after a finite number of $l'$ failure transmissions, where $l'<M$.
In other words, for any $j'\in \{M'+1,\cdots,M\}$, there exists $l'_{j'}<M$ such that the $j'$th column of the matrix $$\text{diag}\{\beta_1,\cdots,\beta_{M'},\underbrace{0,\cdots,0}_{M-M'}\}\mathbf{(DM)}^{l'_{j'}}$$ is not of all zeros and has a positive entry of $\beta'_{j'}$.
Then, we have
\begin{equation}\label{eq:beta_case2}
\begin{aligned}
&\sum_{j=1}^{M}\sum_{k=1}^{M}\sum_{i=1}^{\infty} \rho^{2i}(\mathbf{A}) \left[\text{diag}\{\beta_1,\cdots,\beta_{M'},\underbrace{0,\cdots,0}_{M-M'}\} \right.\\
&\left. \qquad\qquad\qquad\qquad\qquad\qquad  \mathbf{(DM)}^{i-1} \mathbf{(I-D)M}\right]_{j,k}\\
&>\sum_{j=1}^{M}\sum_{k=1}^{M}\sum_{i=l'_{j'}+1}^{\infty} \rho^{2i}(\mathbf{A}) \left[\text{diag}\{\underbrace{0,\cdots,0}_{j'-1},\beta'_{j'},\underbrace{0,\cdots,0}_{M-j'}\} \right. \\
&\left. \qquad\qquad\qquad\qquad\qquad\qquad  \mathbf{(DM)}^{i-l'_{j'}-1} \mathbf{(I-D)M}\right]_{j,k}\\ 
&=\sum_{j=1}^{M}\sum_{k=1}^{M}\sum_{i=1}^{\infty} \rho^{2i}(\mathbf{A}) \left[\text{diag}\{\underbrace{0,\cdots,0}_{j'-1},\rho^{2l'_{j'}}(\mathbf{A})\beta'_{j'},\underbrace{0,\cdots,0}_{M-j'}\} \right. \\
&\left. \qquad\qquad\qquad\qquad\qquad\qquad  \mathbf{(DM)}^{i-1} \mathbf{(I-D)M}\right]_{j,k}.
\end{aligned}
\end{equation}
Applying \eqref{eq:beta_case2} into \eqref{eq:C_case2} for $(M-M')$ times, it can be obtained~as
\begin{equation}
\begin{aligned}
&(M-M'+1)\myexpect{C}\\
&>\eta\sum_{j=1}^{M}\sum_{k=1}^{M} \sum_{i=1}^{\infty} \rho^{2i}(\mathbf{A})\\ 
&\times \left[\text{diag}\{\beta_1,\cdots,\beta_{M'},\underbrace{\rho^{2l'_{M'+1}}(\mathbf{A})\beta'_{M'+1},\cdots,\rho^{2l'_M}(\mathbf{A})\beta'_M}_{M-M'}\}\right. \\
&\left. \qquad\mathbf{(DM)}^{i-1} \mathbf{(I-D)M}\right]_{j,k} \\
&+ (M-M'+1)\gamma''.
\end{aligned}
\end{equation}
Letting $$\underline{\beta} \triangleq \min\{\beta_1,\cdots,\beta_{M'},{\rho^{2l'_{M'+1}}(\mathbf{A})\beta'_{M'+1},\cdots,\rho^{2l'_M}(\mathbf{A})\beta'_M}\}>0$$ and following the same steps in scenario (i), the proof of Theorem~\ref{theorem:main} in scenario (ii) is completed.	
	
    
	\ifCLASSOPTIONcaptionsoff
	\fi

\bibliographystyle{IEEEtran}


\begin{IEEEbiography}
	[{\includegraphics[width=1in,height=1.25in,clip,keepaspectratio]{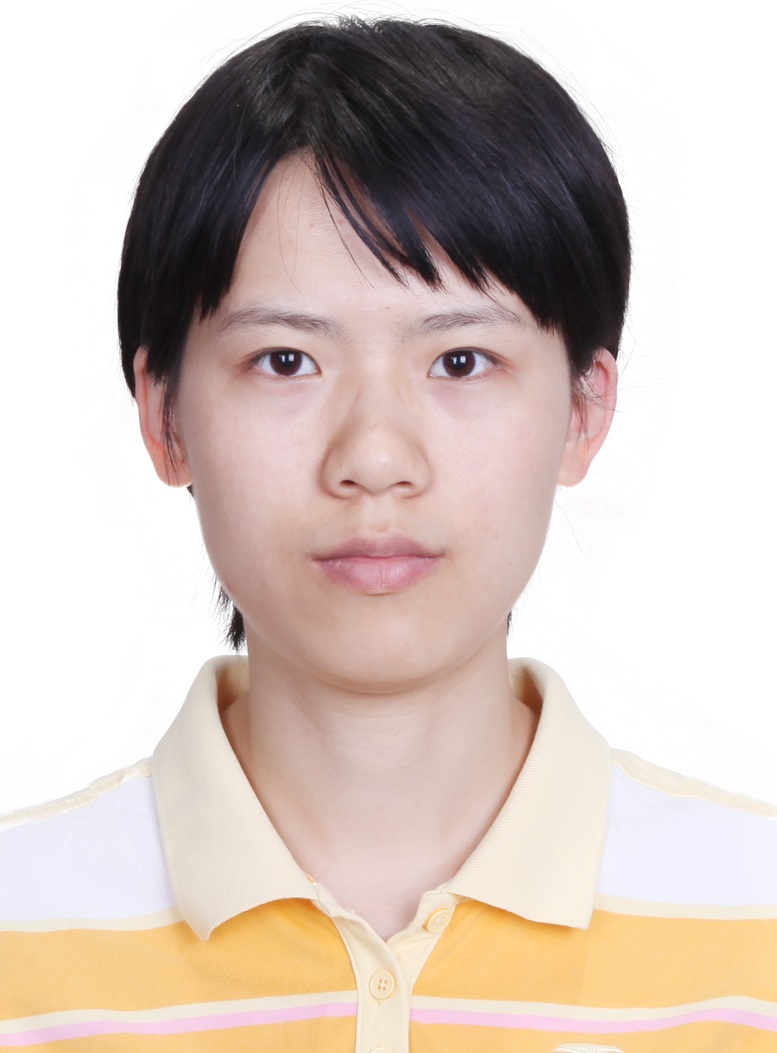}}]{Wanchun Liu}
(M'18) received the B.S. and M.S.E. degrees in electronics and information engineering from Beihang University, Beijing, China, and Ph.D. from The Australian National University, Canberra, Australia. She is currently a Postdoctoral Research Associate at the University of Sydney. Her research interest lies in ultra-reliable low-latency communications and wireless networked control systems for industrial Internet of Things (IIoT). She was a recipient of the Chinese Government Award for Outstanding Students Abroad 2017. She is a co-chair of the Australian Communication Theory Workshop since 2020.
\end{IEEEbiography}
\begin{IEEEbiography}
	[{\includegraphics[width=1in,height=1.25in,clip,keepaspectratio]{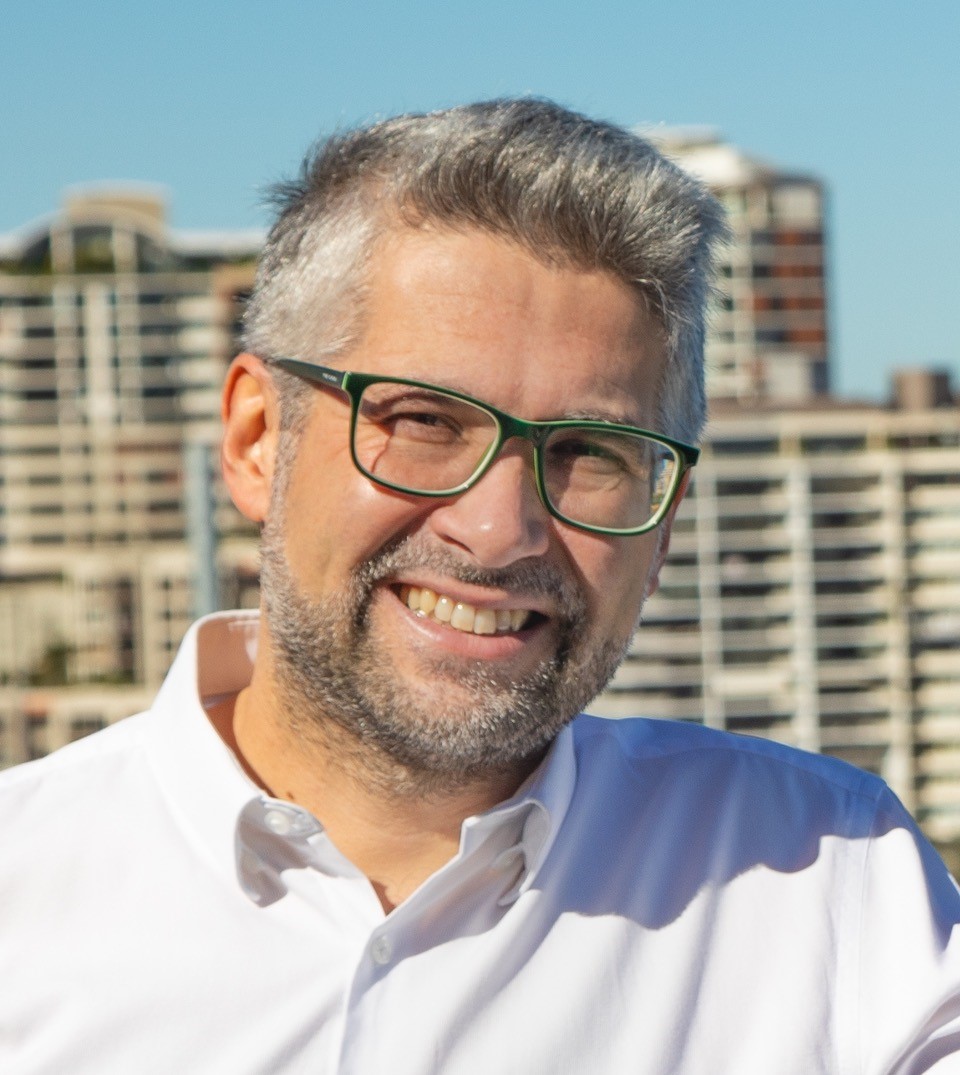}}]{Daniel E. Quevedo}
(S'97--M'05--SM'14--F'21) received Ingeniero Civil Electr\'onico
and M.Sc.\ degrees from   Universidad
T\'ecnica Federico Santa Mar\'{\i}a, Valpara\'{\i}so, Chile, in 2000, and in 2005   the Ph.D.\ degree from the University of Newcastle, Australia.
He is Professor of Cyberphysical Systems  
at the School of Electrical Engineering and Robotics, Queensland University of Technology (QUT), 
in Australia.  
Before joining QUT, he established and led  the Chair in Automatic
Control 
at Paderborn University, Germany. 
In 2003 he 
received the IEEE Conference on Decision and
Control Best Student Paper Award  and was also a finalist  in 
2002.  
He is co-recipient of the  2018 IEEE Transactions on Automatic Control George S.\ Axelby Outstanding Paper Award. 
\par Prof.\ Quevedo currently serves as Associate Editor for   \emph{IEEE Control Systems} and in  the Editorial Board of the \emph{International Journal of Robust and Nonlinear Control}.  From 2015--2018 he was Chair of the IEEE Control Systems Society \emph{Technical Committee on 
	Networks \& Communication Systems}. 
His research interests are in networked control systems, control of power converters and cyberphysical systems security.
\end{IEEEbiography}
\begin{IEEEbiography}
	[{\includegraphics[width=1in,height=1.25in,clip,keepaspectratio]{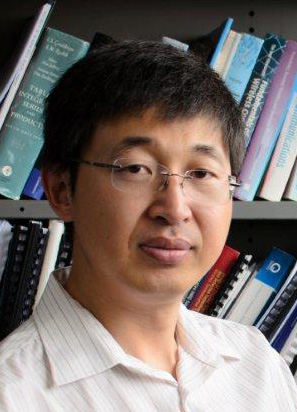}}]{Yonghui Li}
	(M'04-SM'09-F'19) received his PhD degree in November 2002 from Beijing University of Aeronautics and Astronautics. From 1999 – 2003, he was affiliated with Linkair Communication Inc, where he held a position of project manager with responsibility for the design of physical layer solutions for the LAS-CDMA system. Since 2003, he has been with the Centre of Excellence in Telecommunications, the University of Sydney, Australia. He is now a Professor in School of Electrical and Information Engineering, University of Sydney. He is the recipient of the Australian Queen Elizabeth II Fellowship in 2008 and the Australian Future Fellowship in 2012. 
	His current research interests are in the area of wireless communications, with a particular focus on MIMO, millimeter wave communications, machine to machine communications, coding techniques and cooperative communications. He holds a number of patents granted and pending in these fields. He is now an editor for IEEE transactions on communications and IEEE transactions on vehicular technology. He also served as a guest editor for several special issues of IEEE journals, such as IEEE JSAC special issue on Millimeter Wave Communications. He received the best paper awards from IEEE International Conference on Communications (ICC) 2014, IEEE PIMRC 2017 and IEEE Wireless Days Conferences (WD) 2014.	
\end{IEEEbiography}
\begin{IEEEbiography}
	[{\includegraphics[width=1in,height=1.25in,clip,keepaspectratio]{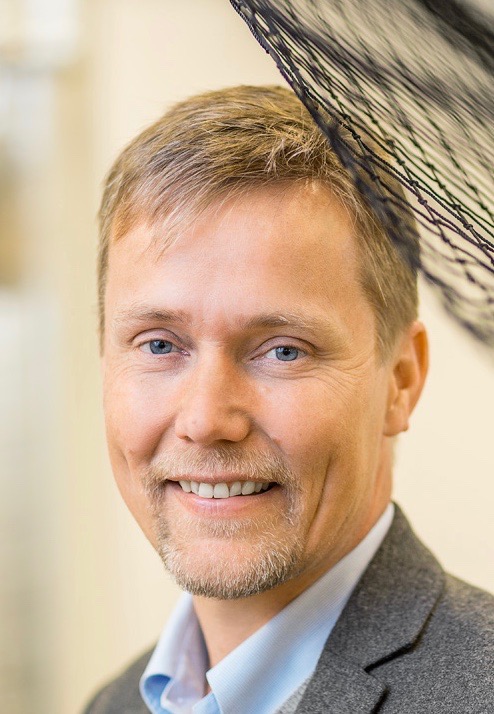}}]{Karl Henrik Johansson} (SM'08-F'13)
is Professor at the School
of Electrical Engineering and Computer Science,
KTH Royal Institute of Technology, Sweden. He received
MSc and PhD degrees from Lund University.
He has held visiting positions at UC Berkeley, Caltech,
NTU, HKUST Institute of Advanced Studies,
and NTNU. His research interests are in networked
control systems, cyber-physical systems, and applications
in transportation, energy, and automation. He
is a member of the Swedish Research Council's Scientific
Council for Natural Sciences and Engineering
Sciences. He has served on the IEEE Control Systems Society Board of
Governors, the IFAC Executive Board, and is currently Vice-President of the
European Control Association Council. He has received several best paper
awards and other distinctions from IEEE, IFAC, and ACM. He has been
awarded Distinguished Professor with the Swedish Research Council and
Wallenberg Scholar with the Knut and Alice Wallenberg Foundation. He has
received the Future Research Leader Award from the Swedish Foundation for
Strategic Research and the triennial Young Author Prize from IFAC. He is
Fellow of the IEEE and the Royal Swedish Academy of Engineering Sciences,
and he is IEEE Control Systems Society Distinguished Lecturer.
	
\end{IEEEbiography}
\begin{IEEEbiography}
	[{\includegraphics[width=1in,height=1.25in,clip,keepaspectratio]{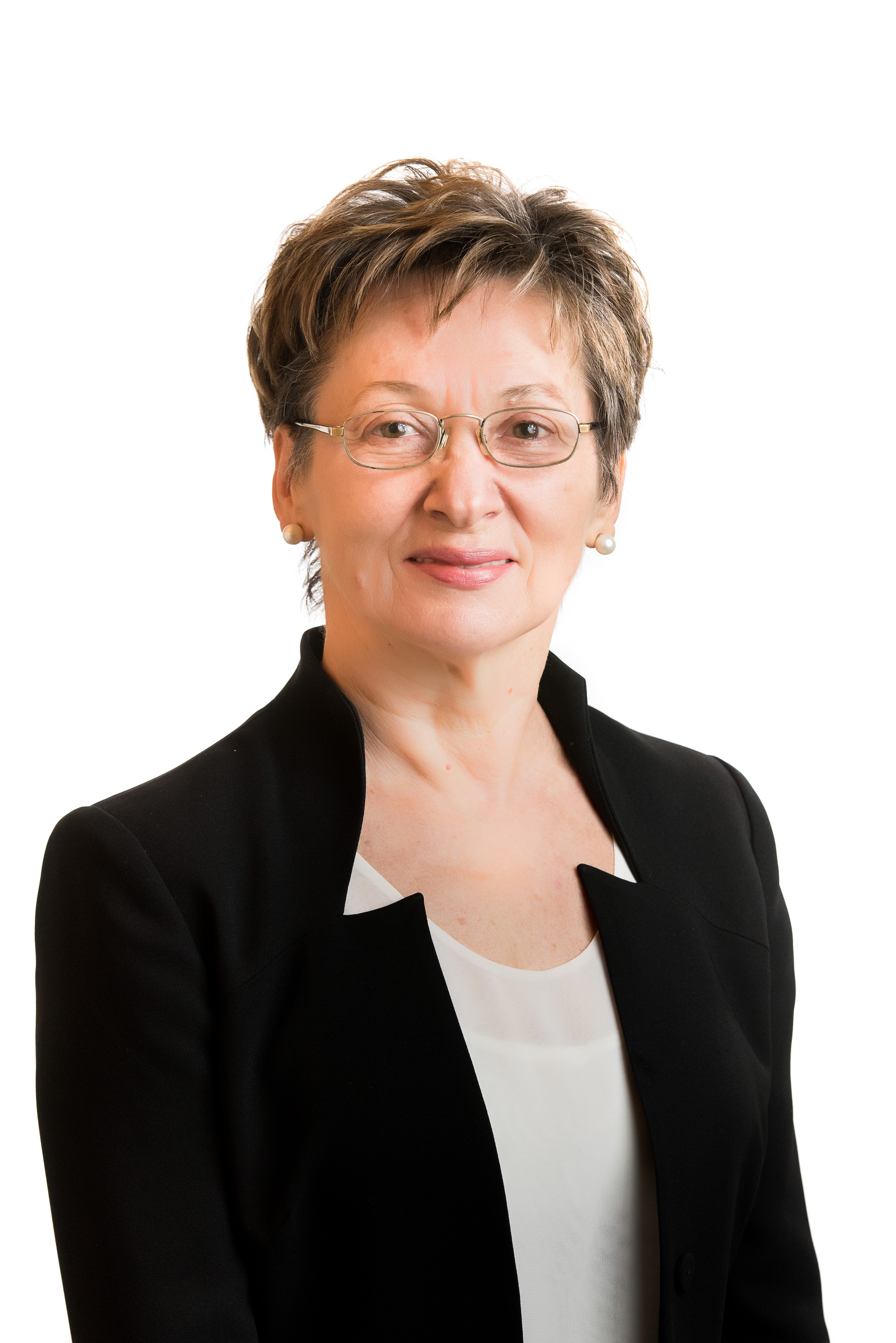}}]{Branka Vucetic}
(F'03) received the BSEE, MSEE and PhD degrees in 1972, 1978 and 1982, respectively, in Electrical Engineering, from The University of Belgrade, Belgrade. She is an ARC Laureate Fellow and Director of the Centre of Excellence for IoT and Telecommunications at the University of Sydney. Her current work is in the areas of wireless networks and Internet of Things. In the area of wireless networks, she explores ultra-reliable, low-latency techniques and transmission in millimetre wave frequency bands. In the area of the Internet of things, Vucetic works on providing wireless connectivity for mission critical applications. Branka Vucetic is a Fellow of IEEE, the Australian Academy of Science and the Australian Academy of Technological Sciences and Engineering.
\end{IEEEbiography}

\end{document}